\documentclass[11pt]{article}
\usepackage{amsmath}
\usepackage{amsthm}
\usepackage{amssymb}
\usepackage{tikz}
\usepackage{float}
\usepackage{diagbox}
\usepackage{enumerate}
\usepackage[margin=1in]{geometry}
\usepackage{hyperref}
\usepackage[utf8]{inputenc}

\let\emph\relax 
\DeclareTextFontCommand{\emph}{\sffamily\em}

\allowdisplaybreaks

\newcommand{\sign}{\mathrm{sign}}
\newcommand{\CSP}{\mathrm{CSP}}
\newcommand{\CH}{\mathrm{Conv}}
\newcommand{\OPT}{\mathrm{OPT}}
\newcommand{\E}{\mathbb{E}}

\theoremstyle{plain}
\newtheorem{theorem}{Theorem}[section]
\newtheorem{lemma}[theorem]{Lemma}
\newtheorem{corollary}[theorem]{Corollary}
\newtheorem{prop}[theorem]{Proposition}

\newtheorem{definition}[theorem]{Definition}
\newtheorem{remark}[theorem]{Remark}
\newtheorem{example}[theorem]{Example}

\title{On the Approximability of Presidential Type Predicates}

\author{Neng Huang\footnote{The University of Chicago. \texttt{nenghuang@uchicago.edu}} \qquad \qquad Aaron Potechin\footnote{The University of Chicago. \texttt{potechin@uchicago.edu}}}
\begin{document}

\maketitle

\begin{abstract}
Given a predicate $P: \{-1, 1\}^k \to \{-1, 1\}$, let $\CSP(P)$ be the set of constraint satisfaction problems whose constraints are of the form $P$. We say that $P$ is approximable if given a nearly satisfiable instance of $\CSP(P)$, there exists a probabilistic polynomial time algorithm that does better than a random assignment. Otherwise, we say that $P$ is approximation resistant.

In this paper, we analyze presidential type predicates, which are balanced linear threshold functions where all of the variables except the first variable (the president) have the same weight. We show that almost all presidential-type predicates $P$ are approximable. More precisely, we prove the following result: for any $\delta_0 > 0$, there exists a $k_0$ such that if $k \geq k_0$, $\delta \in (\delta_0,1 - 2/k]$, and ${\delta}k + k - 1$ is an odd integer then the presidential type predicate $P(x) = \sign({\delta}k{x_1} + \sum_{i=2}^{k}{x_i})$ is approximable. To prove this, we construct a rounding scheme that makes use of biases and pairwise biases. We also give evidence that using pairwise biases is necessary for such rounding schemes.
\end{abstract}

\section{Introduction}
In constraint satisfaction problems (CSPs), we have a set of constraints and we want to satisfy as many of them as possible. Many fundamental problems in computer science are CSPs, including 3-SAT, MAX CUT, $k$-colorability, and unique games.

One fundamental question about CSPs is as follows. For a given type of CSP, is there a randomized polynomial time algorithm which is significantly better than randomly guessing an assignment? More precisely, letting $r$ be the expected proportion of constraints satisfied by a random assignment, is there an $\epsilon > 0$ and a randomized polynomial time algorithm $A$ such that given a CSP instance where at least $(1-\epsilon)$ of the constraints can be satisfied, $A$ returns an $x$ which satisfies at least $(r + \epsilon)$ of the constraints in expectation? If so, we say that this type of CSP is approximable. If not, then we say that this type of CSP is approximation resistant.

For example, H{\aa}stad's 3-bit PCP theorem \cite{hastad_optimal_2001} proves that 3-XOR instances (where every constraint is a linear equation modulo $2$ over $3$ variables) are NP-hard to approximate. A direct corollary of H{\aa}stad's 3-bit PCP theorem is that 3-SAT is also NP-hard to approximate and this theorem has served as the basis for numerous other inapproximability results. On the other hand, Goemans and Williamson's \cite{goemans_improved_1995} breakthrough algorithm for MAX CUT, which gives an approximation ratio of $.878$ for MAX CUT, shows that MAX CUT is approximable as a random cut would only cut half of the edges in expectation.

However, while the approximability or approximation resistance of CSPs has been extensively investigated, there is still much that is unknown. In this paper, we investigate CSPs where every constraint has the form of some fixed presidential type predicate $P$. We show that for almost all presidential type predicates $P$, this type of CSP is approximable.  
\subsection{Definitions}
In order to better describe our results and their relationship to prior work, we need a few definitions.

\begin{definition}
A \emph{boolean predicate} $P$ of arity $k$ is a function $P: \{-1, 1\}^k \to \{-1, 1\}$.
\end{definition}

We remark that in general a predicate can be non-boolean.

\begin{definition}
A \emph{presidential type predicate} is a boolean predicate of the form
\[
P(x_1, \ldots, x_k) = \sign(a\cdot x_1 + x_2 + \cdots + x_k),
\]
where $x_i \in \{-1, 1\}$ for every $i \in [k]$ and $a = a(k)$ is a function on $k$ that takes integer values such that $a + k - 1$ is an odd integer.
\end{definition}

\begin{remark}
In the definition above we assume that $a = a(k)$ takes only integer values. This is not a serious restriction because if $a$ is not an integer, then we can shift $a$ up or down slightly to find another presidential type predicate with integer coefficient $a'$ which is equivalent to the original predicate (see the appendix for a brief proof). We require $a + k - 1$ to be odd in order to prevent a tie.
\end{remark}

We can think of the predicate as a vote where the vote of $x_1$, the ``president", has weight $a$, while votes of remaining voters, the ``citizens", have the same weight 1.

\begin{remark}
Note that presidential type predicates are balanced linear threshold functions, i.e. functions of the form $sign(\sum_{i=1}^{k}{{c_i}x_i})$ where $\forall i, c_i \in \mathbb{R}$ and $\forall x \in \{-1,1\}^k, \sum_{i=1}^{k}{{c_i}x_i} \neq 0$ (so that the function is well-defined). Note that if a predicate $P$ is a balanced linear threshold function, $P(-x) = -P(x)$ so exactly half of the assignments satisfy the predicate and thus a uniformly random assignment has expected value 0.
\end{remark}

\begin{definition}
Given a boolean predicate $P: \{-1, 1\}^k \to \{-1, 1\}$, an instance $\Phi$ of $\CSP(P)$ consists of a set of $n$ variables $x_1, \ldots, x_n$ and $m$ constraints $C_1, \ldots, C_m$ where each $C_i$ has the form
\[
C_i(x_{i_1}, \ldots, x_{i_k}) = P(z_{i, 1}x_{i_1}, \ldots, z_{i, k}x_{i_k})
\]
for some $i_1, \ldots, i_k \in [n]$ and $z_{i, 1} \ldots z_{i, k} \in \{-1, 1\}$.
\end{definition}

\begin{definition}
A boolean predicate $P$ is \emph{approximable} if there exists a constant $\epsilon > 0$ and a polynomial time algorithm, possibly randomized, that on input $\Phi \in \CSP(P)$ such that $\OPT(\Phi) \geq 1 - \epsilon$, produces an assignment to $\Phi$'s variables that in expectation satisfies $r_P + \epsilon$ fraction of the constraints in $\Phi$, where $r_P = \E_{x \in \{-1, 1\}^k}[(1 + P(x))/2]$ is the probability that a constraint in $\Phi$ is satisfied by a random assignment. Otherwise, we say $P$ is \emph{approximation resistant}.

We say that a boolean predicate $P$ is \emph{weakly approximable} if there exists a constant $\epsilon > 0$ and a polynomial time algorithm, possibly randomized, that on input $\Phi \in \CSP(P)$ such that $\OPT(\Phi) \geq 1 - \epsilon$, produces an assignment to $\Phi$'s variables that in expectation either satisfies at least $r_P + \epsilon$ fraction of the constraints in $\Phi$ or satisfies at most $r_P - \epsilon$ fraction of the constraints in $\Phi$. Otherwise, we say that $P$ is \emph{strongly approximation resistant}.
\end{definition}
\begin{remark}
For presidential type predicates, and in fact any odd predicate $P$ (i.e. a predicate $P$ where $P(-x) = -P(x)$), the notions of being approximable and being weakly approximable are equivalent.
\end{remark}

\subsection{Our Results}
In this paper, we prove the following result.

\begin{theorem}\label{thm:2}
For any $\delta_0 > 0$, there exists a $k_0 \in \mathbb{N}$ such that if $k \geq k_0$, $\delta \in (\delta_0,1 - 2/k]$, and ${\delta}k + k - 1$ is an odd integer then the presidential type predicate 
\[
P(x) = \sign\left({\delta}k{x_1} + \sum_{i=2}^{k}{x_i}\right)
\]
is approximable.
\end{theorem}

\begin{remark}
Informally, this theorem says that if the weight of $x_1$ (the ``president'') is at least a constant times $k$, then the predicate is approximable for sufficiently large $k$. We have the condition $\delta \leq 1 - 2/k$ because if $\delta > 1 - 2/k$ and ${\delta}k + k - 1$ is an odd integer, then ${\delta}k \geq k$, which means the predicate is a dictator predicate which is trivially approximable.
\end{remark}

We will prove this theorem by constructing a \textit{rounding scheme} that makes use of \textit{biases} and \textit{pairwise biases}, which are given by a standard semi-definite program (see Section~\ref{standardSDPsubsection}). Complementarily, we also give evidence that using pairwise biases is necessary for such rounding schemes. In particular, we show that for any fixed $\delta > 0$ and degree $m$, for sufficiently large $k$ there is no rounding scheme for the predicate $P(x) = \sign\left({\delta}k{x_1} + \sum_{i=2}^{k}{x_i}\right)$ which has degree at most $m$ and does not use pairwise biases (see Theorem \ref{thm:pairwisenecessity}).

\subsection{Relationship to Prior Work}
We now describe known criteria for determining whether a predicate $P$ is approximable or approximation resistant and how our techniques compare to these criteria.

In 2008, Raghavendra \cite{raghavendra_optimal_2008} gave a characterization of which predicates are approximable and which predicates are approximation resistant. Raghavendra showed that either a standard semi-definite program (SDP) together with an appropriate rounding scheme gives a better approximation ratio than a random assignment or it is unique games hard to do so (see Section \ref{standardSDPsubsection}). However, this characterization leaves much to be desired because for a given predicate, it can be extremely hard to tell which case holds. In fact, it is not even known to be decidable!

Khot, Tulsiani, and Worah \cite{khot_characterization_2013} gave a characterization of which predicates are weakly approximable which is based on whether there exist certain vanishing measures over a polytope which we call the KTW polytope (though similar polytopes were analyzed in some earlier papers, see e.g.~\cite{austrin_randomly_2011, austrin_usefulness_2013,austrin_characterization_2013,karloff_78-approximation_1997}). Unfortunately, it is also unknown whether this characterization is decidable.

Thus, if we want to determine if a given predicate $P$ is approximable or approximation resistant, it is often better to use more direct criteria. For showing that predicates are hard to approximate, the following criterion, proved by Austrin and Mossel \cite{austrin_approximation_2009}, is extremely useful.
\begin{definition}
We say that a predicate $P$ has a balanced pairwise independent distribution of solutions if there exists a distribution $D$ on $\{-1,1\}^k$ such that 
\begin{enumerate}
    \item $D$ is supported on $\{x \in \{-1,1\}^k: P(x) = 1\}$ ($D$ is a distribution of solutions to $P$)
    \item For all $i \in [k]$, $\E_{x \in D}[x_i] = 0$ and for all $i<j \in [k]$, $\E_{x \in D}[{x_i}{x_j}] = 0$
\end{enumerate}
\end{definition}
\begin{theorem}
If $P$ has a balanced pairwise independent distribution of solutions then $P$ is unique games hard to approximate. 
\end{theorem}
This criterion captures most but not all predicates which are known to be unique games hard to approximate. One example of a predicate which is not captured by this criterion is the predicate which was recently constructed by Potechin~\cite{potechin_approximation_2018} which is unique games hard to approximate and is a balanced linear threshold function. \footnote{There were previously known predicates, such as the GLST predicate~\cite{guruswami_tight_1998} $P(x_1,x_2,x_3,x_4) = \frac{1 + x_1}{2}{x_2}{x_3} + \frac{1 - x_1}{2}{x_2}{x_4}$, which are unique games hard (in fact NP-hard) to approximate yet do not have a balanced pairwise independent distribution of solutions. However, the hardness of these predicates can be reduced to the hardness of predicates which do have a balanced pairwise independent distribution of solutions, so Austrin and Mossel's criterion can still be used for these predicates.}

For approximation resistance which does not rely on the hardness of unique games, Chan~\cite{chan_approximation_2016} gave the following stricter criterion which implies NP-hardness of approximation.
\begin{theorem}
If a predicate $P$ has a balanced pairwise independent subgroup of solutions then $P$ is NP-hard to approximate. 
\end{theorem}

For showing that predicates are approximable, the general technique is as follows:
\begin{enumerate}
\item Run Raghavendra's SDP to obtain biases $\{b_i: i \in [n]\}$ and pairwise biases $\{b_{ij}: i < j \in [n]\}$ for the variables.
\item Construct a rounding scheme which takes these biases and pairwise biases and gives us a solution $x$ such that if the SDP ``thinks'' that almost all of the constraints are satisfiable then $x$ satisfies significantly more constraints than a random assignment in expectation.
\end{enumerate}
Based on rounding schemes which are essentially linear in the biases and pairwise biases, Hast~\cite{hast_beating_2005} obtained the following criterion for when predicates are approximable:
\begin{theorem}[Hast's criterion]
Given a predicate $P: \{-1,+1\}^k \to \{-1,+1\}$, 
\begin{enumerate}
\item Define $P_1: \{-1,+1\}^k \to \mathbb{R}$ to be $P_1(x) = \sum_{i=1}^{k}{\hat{P}_{\{i\}}x_i}$
\item Define $P_2: \{-1,+1\}^k \to \mathbb{R}$ to be $P_2(x) =  \sum_{i=1}^{k-1}{\sum_{j=i+1}^{k}{\hat{P}_{\{i,j\}}{x_i}{x_j}}}$.
\end{enumerate}
If there are constants $c_1,c_2$ such that $c_2 \geq 0$ and ${c_1}P_1(x) + {c_2}P_2(x) > 0$ for all $x$ such that $P(x) = 1$ then $P$ is approximable.
\end{theorem}
Aside from Hast's criterion, most of the known approximability results are ad-hoc. Some such results are as follows.
\begin{enumerate}
\item Austrin, Benabbas, and Magen \cite{austrin_quadratic_2010} showed that the monarchy predicate $P(x_1,\cdots,x_k) = \sign((k-2)x_1 + \sum_{i=2}^{k}{x_i})$ is approximable and that any predicate $P$ which is a balanced symmetric quadratic threshold function is approximable. 
\item Potechin~\cite{potechin_approximation_2018} showed that the almost monarchy predicate $P(x_1,\cdots,x_k) = \sign((k-4)x_1 + \sum_{i=2}^{k}{x_i})$ is approximable for sufficiently large $k$.
\end{enumerate}
In this paper, we prove that almost all presidential-type predicates are approximable by generalizing the ideas Potechin~\cite{potechin_approximation_2018} used to prove that the almost monarchy predicate is approximable for sufficiently large $k$ and making these ideas more systematic. Our work compares to previous criteria as follows.

\begin{enumerate}
\item Raghavendra's criterion and the KTW criterion give a space of rounding schemes which should be considered but don't provide an efficient way to search for the best rounding scheme in this space. For our techniques, we take full advantage of this space of rounding schemes while also providing a way to systematically construct the rounding scheme which we need.
\item Like Hast's criterion, we need to check that a certain expression is positive for all $x$ such that $P(x) = 1$. However, there are two key differences between our techniques and Hast's criterion. First, as noted above, we use a larger space of rounding schemes. In particular, we use rounding schemes which are very much non-linear in the biases and pairwise biases. Second, because these rounding schemes are nonlinear in the biases and pairwise biases, it is actually not quite enough to check all $x$ such that $P(x) = 1$. Instead, we need to check over the entire KTW polytope.
\end{enumerate}

\section{Techniques for Analyzing Boolean Predicates}

In this section, we recall techniques for analyzing the approximability of boolean predicates.

\subsection{Fourier Analysis}
In this paper, we will make extensive use of the Fourier expansion of boolean predicates. The Fourier expansion of a $k$-ary boolean predicate $P$ is of the following form
\[
P(x) = \sum_{I \subset [k]}\hat{P}_Ix_I,
\]
where $x_I = \prod_{i \in I}x_i$ and $\{\hat{P}_I: I \subseteq [k]\}$ are the Fourier coefficients $\hat{P}_I = \E_{x \in \{-1,1\}^k}[P(x)x_I]$ of $P$. We have the following lemma for the Fourier coefficients, the proof of which can be found in the appendix.
\begin{lemma}[Fourier coefficients of presidential type predicates]\label{lem:fourier_coeff}
Let $P(x_1, \ldots, x_k) = \sign(a\cdot x_1 + x_2 + \cdots + x_k)$ be a presidential type predicate where $a \leq k - 2$ and $a + k - 1$ is an odd integer. Let $\hat{P}_{tC}$ denote the Fourier coefficient of a set of $t$ citizens (indices from $2$ to $k$) and $\hat{P}_{P+tC}$ denote the Fourier coefficient of a set of $t$ citizens together with the president (index 1). Let $\tau = \lfloor (k - a - 1)/2 \rfloor$. We have
\begin{enumerate}[(1)]
    \item $\hat{P}_P = 1 - \frac{1}{2^{k - 2}}\sum_{l = 0}^\tau\binom{k-1}{l},$
    \item $\hat{P}_{tC} = \frac{1}{2^{k - 2}}\sum_{i = 0}^\tau\sum_{j = 0}^{\tau - i}(-1)^j\binom{k-t-1}{i}\binom{t}{j}, \qquad  \forall t (1 \leq t \leq k - 1 \wedge t \textrm{ is odd}),$
    \item $\hat{P}_{P+tC} = -\frac{1}{2^{k - 2}}\sum_{i = 0}^\tau\sum_{j = 0}^{\tau - i}(-1)^j\binom{k-t-1}{i}\binom{t}{j}, \qquad  \forall t (2 \leq t \leq k - 1 \wedge t \textrm{ is even}).$
\end{enumerate}
\end{lemma}

\subsection{The Standard SDP for CSPs}\label{standardSDPsubsection}
In this subsection we briefly describe Raghavendra's SDP~\cite{raghavendra_optimal_2008}, which together with an appropriate rounding scheme approximates a given CSP in polynomial time as long as the CSP is not unique games hard to approximate. Note that Raghavendra~\cite{raghavendra_optimal_2008} considered CSPs for general constraints but for our discussion here we will focus only on boolean predicates.

We first define the KTW polytope, which plays a crucial role in Khot, Tulsiani and Worah's~\cite{khot_characterization_2013} characterization of which predicates are weakly approximable.
\begin{definition}
Given $x \in \{-1, 1\}^k$, let $p(x) \in \{-1, 1\}^{k + \binom{k}{2}}$ be the vector obtained by concatenating $x$ and $(x_1x_2, x_1x_3, \ldots, x_{k-1}x_k)$. Define
\[
KTW_P = \CH(\{p(x) \mid x \in \{-1, 1\}^k, P(x) = 1\}),
\]
where $P: \{-1, 1\}^k \to \{-1, 1\}$ is a boolean predicate and $\CH(S)$ is the convex hull of $S$. 
\end{definition}

Given a CSP instance on $n$ variables $x_1, \ldots, x_n$ and $m$ constraints $C_1, \ldots, C_m$, the SDP searches for biases $\{b_i\}$ and pairwise biases $\{b_{ij}\}$ whose intended meanings are $b_i = \E[x_i], b_{ij} = \E[x_ix_j]$. Then, for each constraint $C_i$, it searches for a local distribution on the variables in $C_i$ which agrees with the global biases and pairwise biases and maximizes the probability that $C_i$ is satisfied. The goal of the SDP is to find global biases and pairwise biases such that the sum of these satisfying probabilities is maximized.

Let $ALL_k = \CH(\{p(x) \mid x \in \{-1, 1\}^k\})$. For a polytope $A$ and a real number $r$, define $rA = \{rx \mid x \in A\}$. We take the following formal definition of the SDP from~\cite{potechin_approximation_2018}.

\begin{definition}
Let $\Phi$ be a CSP instance on $n$ variables $x_1, \ldots, x_n$ and $m$ constraints $C_1, \ldots, C_m$. The standard SDP for $\Phi$ has the following variables.
\begin{itemize}
    \item $a_{C_i}, p_{C_i, 1}, p_{C_i, 2}$ for each constraint $C_i$.
    \item $b_i$ for each variable $x_i$ and $b_{ij}$ for each pair of variables $x_i, x_j$ where $i < j$.
\end{itemize}
Let $B$ be the square matrix indexed by $\{0, 1, \ldots, n\}$ such that
\begin{itemize}
    \item $B_{ii} = 1$ for every $i \in \{0, 1, \ldots, n\}$.
    \item $B_{0i} = B_{i0} = b_i$ for every $i \in \{1, 2, \ldots, n\}$.
    \item $B_{ij} = B_{ji} = b_{ij}$ for every $i,j \in \{1, 2, \ldots, n\}$ such that $i < j$.
\end{itemize}
The SDP maximizes $\sum_{i = 1}^ma_{C_i}$ subject to the following constraints:
\begin{enumerate}
    \item $B \succeq 0$.
    \item For every constraint $C_i$ on variables $x_{j_1}, \ldots, x_{j_{k_i}}$, where $k_i$ is the arity of $C_i$.
    \begin{enumerate}
        \item $a_{C_i} \in [0, 1]$.
        \item $p_{C_i, 1} \in a_{C_i}KTW_{C_i}$, $p_{C_i, 2} \in (1 - a_{C_i})ALL_{k_i}$.
        \item $p_{C_i, 1} + p_{C_i, 2} = (b_{j_1}, b_{j_2}, \ldots, b_{j_{k_i}}, b_{j_1}b_{j_2}, \ldots, b_{j_{k_i - 1}}b_{j_{k_i}})$.
        
    \end{enumerate}
\end{enumerate}
\end{definition}

Raghavendra showed in~\cite{raghavendra_optimal_2008} that if for all $\epsilon > 0$, this SDP fails to distinguish instances of $\CSP(P)$ where $(1-\epsilon)$ fraction of the constraints are satisfiable from instances of $\CSP(P)$ where at most $(r_P + \epsilon)$ fraction of the constraints are satisfiable, then $P$ is unique games hard to approximate. Conversely, if this SDP does distinguish between these two cases for some $\epsilon > 0$ then $P$ can be approximated by first running this SDP to obtain biases ${b_i: i \in [n]}$ and pairwise biases ${b_{ij}: i < j \in [n]}$ and then applying a rounding scheme of the following form.
\begin{enumerate}
    \item Find unit vectors $\{\vec{v}_i\} \cup {\vec{1}}$ such that for all $i \in [n]$, $\vec{v}_i \cdot \vec{1} = b_i$ and for all $i < j \in [n]$, $\vec{v}_i \cdot \vec{v}_j = b_{ij}$ 
    \item Choose $d$ global vectors $\vec{u}_1,\ldots,\vec{u}_d$ from the multivariate normal distribution.
    \item Set each $x_i = 1$ with probability $p(b_i,\vec{v}_i \cdot \vec{u}_1,\ldots,\vec{v}_i \cdot \vec{u}_d,\vec{1} \cdot \vec{u}_1,\ldots, \vec{1} \cdot \vec{u}_d)$ for some function $p:\mathbb{R}^{2d+1} \to [0,1]$ (which is the same for all $i$).
\end{enumerate}
\begin{example}
In hyperplane rounding, we choose a single global vector $\vec{u}$ and set $x_i = sign(\vec{v}_i \cdot \vec{u})$.
\end{example}

\subsection{Choosing Rounding Schemes}
While thinking about rounding schemes in terms of vectors is general (assuming the unique games conjecture or at least that unique games is hard), it is rather unwieldy. Instead, we think of rounding schemes in terms of the expected value $\E[x_I]$ of each monomial $x_I = \prod_{i \in I}{x_i}$. Thus, we choose our rounding schemes by choosing $\E[x_I]$ for each monomial $x_I$. However, we do not have complete freedom for these choices. Intuitively, $\E[x_I]$ should obey the following constraints:
\begin{enumerate}
    \item $\E[x_I]$ is a function of $\{b_i \mid i \in I\}$ and $\{b_{ij} \mid i, j \in I\}$.
    \item $\E[x_I]$ is invariant under permutation of indices in I.
    \item If we flip the sign of any variable $x_i$ where $i \in I$ (by flipping the signs of $b_i$ and $\{b_{ij}: j \in I, j \neq i\}$), then the sign of $\E[x_I]$ should be flipped as well.
\end{enumerate}

It turns out that for determining whether a predicate $P$ is \emph{weakly approximable} (which is the same as approximable for presidential type predicates), these are the only constraints on $\E[x_I]$. More precisely, we have the following theorem from~\cite{potechin_approximation_2018}, which is also implicit in \cite{khot_characterization_2013}:

\begin{theorem}[Theorem 5.1 in~\cite{potechin_approximation_2018}]\label{thm:rounding}
Let $\{b_i \mid i \in [k]\}$ and $\{b_{ij} \mid i, j \in [k], i < j\}$ be biases and pairwise biases produced by the standard SDP. For every $a \in [k]$, let $f_a : [-1, 1]^{a + \binom{a}{2}} \to [-1, 1]$ be a continuous function satisfying the following symmetric requirements.
\begin{enumerate}
    \item For all permutations $\sigma \in S_a$,
    \[
    f_a(b_{i_{\sigma(1)}}, \ldots, b_{i_{\sigma(a)}}, b_{i_{\sigma(1)}i_{\sigma(2)}}, \ldots, b_{i_{\sigma(a-1)}i_{\sigma(a)}}) = 
    f_a(b_{i_1}, \ldots, b_{i_a}, b_{i_1i_2}, \ldots, b_{i_{a-1}i_a}).
    \]
    \item For all signs $s_{i_1}, \ldots, s_{i_a} \in \{-1, 1\}^a$,
    \[
    f_a(s_{i_1}b_{i_1}, \ldots, s_{i_a}b_{i_a}, s_{i_1}s_{i_2}b_{i_1i_2}, \ldots, s_{i_{a-1}}s_{i_a}b_{i_{a-1}i_a}) = \left(\prod_{j = 1}^as_{i_j}\right)
    f_a(b_{i_1}, \ldots, b_{i_a}, b_{i_1i_2}, \ldots, b_{i_{a-1}i_a}).
    \]
\end{enumerate}
Then there exists a sequence of rounding schemes $\{R_q\}$ and coefficients $\{c_q\}$ such that for all subsets $I = \{i_1, \ldots, i_a\}$ of size at most $k$, 
\[
\sum_qc_q\E_{R_q}[x_I] = f_a(b_{i_1}, \ldots, b_{i_a}, b_{i_1i_2}, \ldots, b_{i_{a-1}i_a}),
\]
where $\E_{R_q}[x_I]$ is the expected value of $x_I$ given by rounding scheme $R_q$. Moreover, this sum can be taken to be globally convergent.
\end{theorem}
\begin{remark}\label{remark:rounding_issues}
This theorem gives us a linear combination of rounding schemes. The coefficients $c_q$ can be thought of as a probability distribution of rounding schemes, but there are two problems:
\begin{itemize}
    \item $\sum_{q}|c_q|$ may not be 1. One fix to this issue is to scale $f$ by an appropriate constant $\epsilon$.
    \item $c_q$ may be negative. In general, this can be a real issue but here the predicates we consider are odd, which means if $c_q$ is negative we can simply flip the rounding scheme $R_q$ and take it with probability $-c_q$.
\end{itemize}
\end{remark}
\begin{example}
This theorem says the following about $\E[x_i]$ and $\E[{x_i}{x_j}]$.
\begin{itemize}
    \item We can take $\E[x_i] \sim f_1(b_i)$ for any continuous function $f_1$ such that $f_1(b_i) = -f_1(-b_i)$ (i.e. $f_1$ is odd).
    \item We can take $\E[{x_i}{x_j}] \sim f_2(b_i,b_j,b_{ij})$ for any continuous function $f_2$ such that $f_2(b_i, b_j, b_{ij}) = f_2(b_j, b_i, b_{ij}) = -f_2(-b_i, b_j, -b_{ij})$. The first equality corresponds to exchanging $i$ and $j$ while the second equality corresponds to flipping $x_i$.
\end{itemize}
\end{example}
\begin{example}\label{rounding_example}
Some examples of possible functions $f_3$ are as follows: 
\begin{enumerate}
    \item We can take $\E[{x_i}{x_j}{x_k}] \sim {x_i}{x_j}{x_k}$
    \item As discussed in the following subsections, we will take $\E[{x_i}{x_j}{x_k}] \sim ({b_i}b_{jk} + {b_j}b_{ik} + b_{k}b_{ij})$
    \item Potechin~\cite{potechin_approximation_2018} found a simpler rounding scheme for the monarchy predicate where 
    $\E[{x_i}{x_j}{x_k}] \sim sign({x_i}{x_j}{x_k})\max\{|x_i|,|x_j|,|x_k|\}$
\end{enumerate}
\end{example}
In choosing the rounding scheme, our goal is as follows. For each constraint, the standard SDP could give us any point in the KTW polytope. We need to show that no matter which point in the KTW polytope we are given, the probability that the rounding scheme satisfies the constraint is better than a random guess. Equivalently, we need to show that for all points in the KTW polytope,
\[
\sum_{I \subseteq [k]:I \neq \emptyset}{\hat{P}_{I}\E[x_I]} > 0
\]
\begin{example}
Consider the majority predicate $P(x_1,\ldots,x_k) = \sign(x_1 + \ldots + x_k)$. If we take $E[x_i] = f_1(b_i) = {\epsilon}b_i$ and take $f_a = 0$ whenever $a > 1$ then 
\[
\sum_{I \subseteq [k]:I \neq \emptyset}{\hat{P}_{I}\E[x_I]} = \epsilon{\hat{P}_{\{1\}}}\sum_{i=1}^{k}{b_i}
\]
Since $\sum_{i=1}^{k}{x_i} \geq 1$ for every satisfying assignment, for any point in the KTW polytope, $\sum_{i=1}^{k}{b_i} \geq 1$ and thus $\sum_{I \subseteq [k]:I \neq \emptyset}{\hat{P}_{I}\E[x_I]} \geq {\epsilon}\hat{P}_{\{1\}} > 0$.
\end{example}
\section{Techniques for Approximating Presidential Type Predicates}
In this section, we describe our techniques for approximating presidential type predicates. These techniques are a generalization of the techniques used in \cite{potechin_approximation_2018} to show that the almost monarchy predicate is approximable for sufficiently large $k$.
\subsection{High Level Overview}\label{overviewsubsection}
To approximate the presidential type predicate $P(x) = \sign\left({\delta}k{x_1} + \sum_{i=2}^{k}{x_i}\right)$, we use the following type of rounding scheme.
\begin{enumerate}
    \item $f_1(b_i) = {c_1}b_i$.
    \item $f_{2l+1}(b_{i_1},\ldots,b_{i_{2l+1}},b_{{i_1}{i_2}},\ldots,
    b_{{i_{2l}}{i_{2l+1}}}) = 
    c_{2l + 1}\left(b_{i_1}b_{i_2i_3}\cdots b_{i_{2l}i_{2l+1}} + \text{symmetric terms}\right)$ 
\end{enumerate}
where we need to carefully choose the coefficients $c_1,c_3,\ldots$ so that for all points in the KTW polytope,
\[
\sum_{I \subseteq [k]:I \neq \emptyset}{\hat{P}_{I}\E[x_I]} > 0
\]
Because of the symmetry of presidential type predicates $P$, we can analyze $\sum_{I \subseteq [k]:I \neq \emptyset}{\hat{P}_{I}\E[x_I]}$ in terms of a few key functions of the biases and pairwise biases.
\begin{definition}\label{def:keyquantities}
Given biases $\{b_i: i \in [k]\}$ and pairwise biases $\{b_{ij}: i < j \in [k]\}$, we make the following definitions:
\begin{enumerate}
    \item We define $\alpha = b_1$
    \item We define $\beta = \sum_{i=2}^{k}{b_i}$
    \item We define $S_{\{\{i_1,i_2\}\}} = \sum_{i < j \in [k]}{b_{ij}}$. We then write $S_{\{\{i_1,i_2\}\}} = E(1 + \Delta)$ where $E = \frac{\delta^2k^2}{2} - \frac{k}{2} + 1$ is the value we expect for $S_{\{\{i_1,i_2\}\}}$ and $\Delta$ measures how far $S_{\{\{i_1,i_2\}\}}$ is from this expected value.
\end{enumerate}
\end{definition}
With these definitions, we can approximate $\sum_{I \subseteq [k]:I \neq \emptyset}{\hat{P}_{I}\E[x_I]}$ in terms of $\alpha$, $\beta$, and $\Delta$. Our strategy is now as follows:
\begin{enumerate}
    \item We choose a polynomial $h(x) = \sum_{l = 1}^{m}{{a_l}x^l}$ so that $h(1 + \Delta) \approx 1$ except near $\Delta = -1$ as we must have that $h(0) = 0$. More precisely, we choose $h$ to satisfy certain properties (see Lemma~\ref{lem:hconditions}).
    
    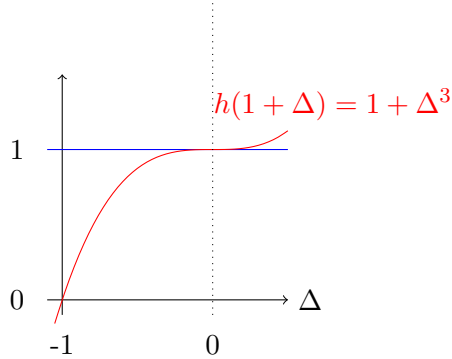
\begin{figure}[H]
\centering
\begin{tikzpicture}[scale = 2]
      \draw[->] (-0.1,0) -- (1.5,0) node[right] {$\Delta$};
      \draw[->] (0,-0.1) -- (0,1.5); 
      \draw[domain=-0.1:1.5,smooth,variable=\x,blue] plot ({\x},{1});
      \draw[domain=-0.05:1.5,smooth,variable=\x,red] plot ({\x},{3*\x-3*\x*\x+\x*\x*\x});
      \draw[dotted] (1, -0.1) -- (1, 2);
      \node[red] at (1.8, 1.3) {$h(1 + \Delta) = 1 + \Delta^3$};
      \node at (-0.3, 1) {$1$};
      \node at (-0.3, 0) {$0$};
      \node at (1, -0.3) {$0$};
      \node at (0, -0.3) {-$1$};
\end{tikzpicture}
\caption{Plot of $h(1 + \Delta) = 1 + \Delta^3$.}
\end{figure}
\begin{remark}
A reasonably good choice for $h$ is $h(1 + \Delta) = 1 + {\Delta}^3$, which was used to give an approximation algorithm for the almost monarchy predicate for sufficiently large $k$ \cite{potechin_approximation_2018}. In fact, while we don't prove it here, for quasi-monarchy predicates of the form $P(x_1,\ldots,x_k) = sign\left((k-2c)x_1 + \sum_{i=2}^{k}{x_i}\right)$ for a fixed constant $c$, $h(1 + \Delta) = 1 + {\Delta}^3$ is sufficient to give an approximation algorithm for sufficiently large $k$. However, this $h$ is not sufficient to give an approximation algorithm for more general presidential type predicates because $h(1+ \Delta) = 1 + {\Delta}^3$ is far from $1$ if $\Delta$ is much larger than $0$.
\end{remark}
    \item We choose the coefficients $\{c_1\} \cup \{c_{2l+1}: l \in [m]\}$ so that 
    \begin{align*}
    \sum_{I \subseteq [k]:I \neq \emptyset}{\hat{P}_{I}\E[x_I]} &= \left({\delta}k^2 + \frac{k}{\delta}\right)\alpha + k\left(\beta - \frac{\alpha}{\delta}\right)h(1 + \Delta) + O(k)\cdot \Delta + O(1) \\
    &= k({\delta}k\alpha + \beta) + k\left(\beta - \frac{\alpha}{\delta}\right)(h(1 + \Delta) - 1) + O(k)\cdot \Delta + O(1)
    \end{align*}
    \item Since for every satisfying assignment, ${\delta}k{x_1} + \sum_{i=2}^{k}{x_i} \geq 1$, for every point in the KTW polytope, 
    \[
    {\delta}k{b_1} + \sum_{i=2}^{k}{b_i} = {\delta}k{\alpha} + \beta \geq 1
    \]
    and thus $k({\delta}k\alpha + \beta) \geq k$. If we could show that the remaining terms $k\left(\beta - \frac{\alpha}{\delta}\right)(h(1 + \Delta) - 1) + O(k)\cdot \Delta + O(1)$ are $o(k)$, then we would be done. Unfortunately, this may not be true when $|\Delta|$ is large.
    \item To handle this, we show that if $|\Delta|$ is large then we can obtain a considerably better bound on ${\delta}k\alpha + \beta$ More precisely, we proceed as follows:
    \begin{enumerate}
        \item When $\Delta \geq -0.55$, we show that 
        $
        \delta k\alpha + \beta \geq \frac{(\delta^2 k - 1)|\Delta|}{4} + \frac{1}{2}
        $
        (see Lemma \ref{lem:linear2}). As long as $h(1+\Delta)$ is sufficiently close to $1$, this allows us to show that $k({\delta}k\alpha + \beta) + k\left(\beta - \frac{\alpha}{\delta}\right)(h(1 + \Delta) - 1) + O(k)\cdot \Delta + O(1)$ is positive.
        \item When $\Delta < -0.55$, we show that we must have $\alpha > 0$. In this case, we rewrite $\sum_{I \subseteq [k]:I \neq \emptyset}{\hat{P}_{I}\E[x_I]} = \left({\delta}k^2 + \frac{k}{\delta}\right)\alpha + k\left(\beta - \frac{\alpha}{\delta}\right)h(1 + \Delta) + O(k)\cdot \Delta + O(1)$ as
        \[
        k(\delta k\alpha + \beta)h(1 + \Delta) + \left(\delta k^2 + \frac{k}{\delta}\right)\alpha(1 - h(1 + \Delta)) + O(k)\cdot\Delta + O(1).
        \]
        and show that the sum of the first two terms is positive and $\Omega(k^2)$.
    \end{enumerate}
\end{enumerate}
\subsection{Sums of Products of Biases and Pairwise Biases}
In order to implement this strategy, we need some notations related to biases and pairwise biases. Note that similar definitions were also used in~\cite{potechin_approximation_2018}.

\begin{definition}
For $E_1 \subseteq [k]$ and $E_2 \subseteq \binom{[k]}{2}$, define
\[
B_{E_1, E_2} = \prod_{i \in E_1} b_i \prod_{\substack{\{i, j\} \in E_2 \\ i < j}}b_{ij}.
\]
\end{definition}

\begin{definition}
Let $V = \{\alpha, i_1, i_2, \ldots, i_{k-1}\}$. Let $H = H_1 \cup H_2$ where $H_1 \subseteq V$ and $H_2 \subseteq \binom{V}{2}$. Define
\[
S_H = \sum_{\substack{E_1, E_2: \exists \sigma: V \to [k] \text{ bijective}\\ \sigma(\alpha) = 1, \sigma(H_1) = E_1, \sigma(H_2) = E_2}}B_{E_1, E_2},
\]
where $\sigma(H_1) = \{\sigma(i) \mid i \in H_1\}$, $\sigma(H_2) = \{\{\sigma(i), \sigma(j)\} \mid \{i, j\} \in H_2\}$.
\end{definition}

Intuitively, $S_H$ is the sum of products $B_{E_1, E_2}$ where $E_1\cup E_2$ has the form $H$. One particularly important such sum in our algorithms is $S_{\{\{i_1, i_2\}\}}$, which is the sum of pairwise biases with indices in $[2, k]$. 

\begin{definition}\label{def:shorthand_sums}
We define the following shorthand notations for some important sums.
\begin{align*}
S_{1, l} &= S_{\{i_1, \{i_2, i_3\}, \{i_4, i_5\}, \ldots, \{i_{2l}, i_{2l+1}\}\}}, \\
S_{2, l} &= S_{\{\alpha, \{i_1, i_2\}, \{i_3, i_4\}, \ldots, \{i_{2l-1}, i_{2l}\}\}}, \\
S_{3, l} &= S_{\{i_1, \{\alpha, i_2\}, \{i_3, i_4\}, \ldots, \{i_{2l-1}, i_{2l}\}\}}.
\end{align*}
\end{definition}

\begin{example}
In the case where $k = 4, l = 1$, we have
\begin{align*}
S_{1, l} &= b_2b_{34} + b_3b_{24} + b_4b_{23}, \\
S_{2, l} &= b_1b_{23} + b_1b_{24} + b_1b_{34}, \\
S_{3, l} &= b_2b_{13} + b_2b_{14} + b_3b_{12} + b_3b_{14} + b_4b_{12} + b_4b_{13}. \\
\end{align*}
\end{example}
The reason that these sums are important is because they are the main terms which appear when we evaluate $\sum_{I \subseteq [k]:I \neq \emptyset}{\hat{P}_{I}\E[x_I]}$. 
\begin{prop}
If we take 
\[
f_{2l+1}(b_{i_1},\ldots,b_{i_{2l+1}},b_{{i_1}{i_2}},\ldots,
    b_{{i_{2l}}{i_{2l+1}}}) = 
    c_{2l + 1}\left(b_{i_1}b_{i_2i_3}\cdots b_{i_{2l}i_{2l+1}} + \text{symmetric terms}\right)
\]
then
\[
\sum_{|I| = 2l+1}{\hat{P}_I \E[X_I]} \\
 = c_{2l+1}\left(\hat{P}_{(2l+1)C}S_{1, l} + \hat{P}_{P+(2l)C}(S_{2, l} + S_{3, l})\right)
 \]
\end{prop}
To approximate these sums, we use the following proposition (recall that we set $S_{\{\{i_1,i_2\}\}} = E(1 + \Delta)$). The proof of this proposition can be found in the appendix.
\begin{prop}\label{prop:bias_sum} 
For every $l \geq 1$,
\begin{align*}
    \frac{l!}{E^l}S_{1, l} &= \beta(1 + \Delta)^l - \frac{S_{\{i_1, \{i_1, i_2\}\}}}{E}l(1 + \Delta)^{l - 1} - \frac{\beta S_{\{\{i_1, i_2\}, \{i_1, i_3\}\}}}{E^2}l(l-1)(1 + \Delta)^{l - 2} + O\left(\frac{1}{k}\right),\\
    \frac{l!}{E^l}S_{2, l} &= \alpha(1 + \Delta)^l + O\left(\frac{1}{k}\right),\\
    \frac{l!}{E^l}S_{3, l} &= \frac{\beta S_{\{\{\alpha, i_1\}\}}}{E}l(1 + \Delta)^{l-1} + O\left(\frac{1}{k}\right),\\
\end{align*}
where the hidden constants in big-O may depend on $l$.
\end{prop}
\section{Proof of Theorem~\ref{thm:2}}
In this section, we prove Theorem~\ref{thm:2}.  

{
\renewcommand{\thetheorem}{\ref{thm:2}}

\begin{theorem}[Restated]
For any $\delta_0 > 0$, there exists a $k_0 \in \mathbb{N}$ such that if $k \geq k_0$, $\delta \in (\delta_0,1 - 2/k]$, and ${\delta}k + k - 1$ is an odd integer then the presidential-type predicate 
\[
P(x) = \sign\left({\delta}k{x_1} + \sum_{i=2}^{k}{x_i}\right)
\]
is approximable.
\end{theorem}
\addtocounter{theorem}{-1}
}
In particular, we prove that for sufficiently large $k$ and a carefully chosen polynomial $h$, the following rounding scheme approximates the presidential type predicate $P(x) = \sign\left({\delta}k{x_1} + \sum_{i=2}^{k}{x_i}\right)$.
\begin{definition}
Given a polynomial $h(x) = \sum_{l = 1}^{m}{{a_l}x^l}$, we define $R_{k,\delta,h}$ to be the rounding scheme such that setting $u = \frac{1+\delta}{2}k$, $v = \frac{1-\delta}{2}k$, and $E = \frac{\delta^2k^2}{2} - \frac{k}{2} + 1$,
\begin{enumerate}
    \item $f_1(b_i) = \left({\delta}k^2 + \frac{k}{\delta}\right)b_i$
    \item For all $l \in [m]$, 
    \[
    f_{2l+1}(b_{i_1},\ldots,b_{i_{2l+1}},b_{{i_1}{i_2}},\ldots,
    b_{{i_{2l}}{i_{2l+1}}}) = 
    c_{2l + 1}\left(b_{i_1}b_{i_2i_3}\cdots b_{i_{2l}i_{2l+1}} + \text{symmetric terms}\right)
    \]
    where $c_{2l+1} = a_l \cdot \frac{2^{k-2}(u-1)!(v-1)!}{(k - 2l - 2)!\delta^{2l}k^{2l-1}E^l}$, .
\end{enumerate}
\end{definition}
\begin{theorem}\label{thm:3}
For all $\delta_0 > 0$, if $h = \sum_{l = 1}^{m}{{a_l}x^l}$ is a polynomial such that 
\begin{enumerate}
    \item $h'(1) = h''(1) = 0$,
    \item For all $\Delta \in [-0.55,\frac{1}{\delta_0^2}]$, $|h(1+\Delta)-1| \leq \frac{{\delta^2_0}|\Delta|}{5}$,
    \item For all $\Delta \in [-1,-0.55]$, $0 \leq h(1+\Delta) \leq 1$,
\end{enumerate}
then there exists a $k_0 \in \mathbb{N}$ such that for all $\delta \geq \delta_0$ and $k \geq k_0$ where ${\delta}k + k - 1$ is an odd integer, $R_{k,\delta,h}$ approximates the presidential type predicate P(x) = $\sign\left({\delta}k{x_1} + \sum_{i=2}^{k}{x_i}\right)$.
\end{theorem}

\begin{remark}
As described in Section ~\ref{overviewsubsection}, our proof contains a case analysis of $\Delta$. The value $-0.55$ is chosen because when $\Delta < -0.55$, the bias $\alpha$ of the president is always positive.
\end{remark}

This section is organized as follows. We first compute the expected value of the rounding scheme in terms of $h$. Then, we show that if $h$ has the required properties, then the expected value is positive over the entire polytope, which implies that our predicate is approximable. Finally, we find such a polynomial with the desired properties. 

\subsection{Evaluating the Rounding Scheme}
In this subsection, we analyze $\sum_{I \subseteq [k]:I \neq \emptyset}{\hat{P}_{I}\E[x_I]}$ in terms of $h$. 

We have the following lemma for Fourier coefficients, the proof of which can be found in the appendix.
\begin{lemma}\label{lem:fourier_coeff_2}
Let $\delta_0 > 0$ be a constant. Let $P(x_1, \ldots, x_k) = \sign(\delta\cdot k x_1 + x_2 + \cdots + x_k)$ where $\delta \in [\delta_0, 1)$ such that $\delta k + k - 1$ is an odd integer. Let $u = \frac{1+\delta}{2}k$ and $v = \frac{1-\delta}{2}k$. Let $\hat{P}_{tC}$ denote the Fourier coefficient of a set of $t$ citizens and $\hat{P}_{P+tC}$ denote the Fourier coefficient of a set of $t$ citizens together with the president. We have
\begin{align*}
    &\hat{P}_P = 1 - \frac{1}{2^{k - 2}}\sum_{l = 0}^{v - 1}\binom{k-1}{l}, \\
&\hat{P}_{tC} = \frac{1}{2^{k - 2}}\cdot\frac{(k - t - 1)!}{(u-1)!(v-1)!}\left(\delta^{t-1}k^{t-1} - \frac{(t-1)(t-2)}{2}\delta^{t-3}k^{t - 2} + O(k^{t - 3})\right), \quad t \textrm{ is an odd constant}\\
    &\hat{P}_{P+tC} = -\frac{1}{2^{k - 2}}\cdot\frac{(k - t - 1)!}{(u-1)!(v-1)!}\left(\delta^{t-1}k^{t-1} - \frac{(t-1)(t-2)}{2}\delta^{t-3}k^{t - 2} + O(k^{t - 3})\right), \quad t \textrm{ is an even constant}\\
\end{align*}
where the constants inside the big $O$s grows with $t$ but not with $\delta$.
\end{lemma}
\begin{remark}
The lemma allows $\delta$ to depend on $k$ as long as $\delta = \Omega(1)$. In particular, we can take $\delta = 1-\frac{2c}{k}$ for any constant $c \geq 1$. Also, when $\delta$ is at least a constant we have that $\hat{P}_P$ is exponentially larger than $\hat{P}_C$.
\end{remark}
Recall that we set $S_{\{\{i_1, i_2\}\}} = E(1 + \Delta)$ where $E = \frac{\delta^2k^2}{2} - \frac{k}{2} + 1$ (see Definition~\ref{def:keyquantities}). The reason for this choice for $E$ is as follows. We expect that the cases which are most difficult to round are the two cases where $\delta k \alpha + b = 1$:
\begin{itemize}
    \item The president and $\frac{1 - \delta}{2}k$ citizens vote 1, others vote $-1$. In this case,
    \[
    \sum_{i < j \in [2,k]}x_ix_j = \binom{k-1}{2} - 2\frac{(1 - \delta)k}{2}\left(\frac{(1 + \delta)k}{2} - 1\right)
    \]
    \item The president and $\frac{1 + \delta}{2}k$ citizens vote $-1$, others vote 1. In this case, 
    \[
    \sum_{i < j \in [2,k]}x_ix_j = \binom{k-1}{2} - 2\frac{(1 + \delta)k}{2}\left(\frac{(1 - \delta)k}{2} - 1\right)
    \]
\end{itemize}
For both of these cases, $\sum_{i < j \in [2,k]}x_ix_j$ is approximately $\frac{\delta^2k^2}{2}$. Taking the average of these two cases we have $E = \frac{\delta^2k^2}{2} - \frac{k}{2} + 1$. Note that since $\delta > \delta_0$ is at least a constant, we have $E = \Omega(k^2)$.

We now analyze $\sum_{I \subseteq [k]:I \neq \emptyset}{\hat{P}_{I}\E[x_I]}$ in terms of $h$.
\begin{lemma}
Assume that we have $h(x) = \sum_{l = 1}^{m}{{a_l}x^l}$ and coefficients 
\[
c_{2l+1} = a_l \cdot \frac{2^{k-2}(u-1)!(v-1)!}{(k - 2l - 2)!\delta^{2l}k^{2l-1}E^l}
\]
where $u = \frac{1+\delta}{2}k$, $v = \frac{1-\delta}{2}k$, and $E = \frac{\delta^2k^2}{2} - \frac{k}{2} + 1$. The contribution of degree $\geq 3$ terms is
\begin{align*}
&k\left(\beta - \frac{\alpha}{\delta}\right)h(1 + \Delta) - \frac{2(1 + \Delta)^2\beta h''(1+\Delta)}{\delta^2} + \frac{(1 + \Delta)\beta h'(1+\Delta)}{\delta^2} \\
& \qquad - k\left(\frac{S_{\{i_1, \{i_1, i_2\}\}}}{E}h'(1 + \Delta) + \frac{\beta S_{\{\{i_1, i_2\}, \{i_1, i_3\}\}}}{E^2}h''(1 + \Delta) + \frac{\beta S_{\{\{\alpha, i_1\}\}}}{E\delta}h'(1 + \Delta)\right) + O(1)
\end{align*}

\end{lemma}
\begin{proof}
We have the following computation:
\begin{align*}
&\qquad\sum_{l = 1}^m\sum_{|I| = 2l+1}\hat{P}_I \E[X_I] \\
& = \sum_{l = 1}^mc_{2l+1}\left(\hat{P}_{(2l+1)C}S_{1, l} + \hat{P}_{P+(2l)C}(S_{2, l} + S_{3, l})\right) \\
& = \sum_{l = 1}^m a_l\Bigg(\left(k - l(2l-1)\delta^{-2}+O\left(\frac{1}{k}\right)\right)\cdot\\
& \qquad\left(\beta(1 + \Delta)^l - \frac{S_{\{i_1, \{i_1, i_2\}\}}}{E}l(1 + \Delta)^{l - 1} - \frac{\beta S_{\{\{i_1, i_2\}, \{i_1, i_3\}\}}}{E^2}l(l-1)(1 + \Delta)^{l - 2} + O\left(\frac{1}{k}\right)\right) - \\
& \qquad \left(1 - \frac{2l+1}{k}\right)\cdot\left(\frac{k}{\delta}- \frac{(l-1)(2l-1)}{\delta^3}+O\left(\frac{1}{k}\right)\right)\cdot\left(\alpha(1 + \Delta)^l + \frac{\beta S_{\{\{\alpha, i_1\}\}}}{E}l(1 + \Delta)^{l-1} + O\left(\frac{1}{k}\right)\right)\Bigg) \\
& = \sum_{l = 1}^m a_l\Bigg(k\left(\beta - \frac{\alpha}{\delta}\right)(1 + \Delta)^l - \frac{l(2l-1)\beta(1+\Delta)^l}{\delta^2} \\
& \qquad - k\left(\frac{S_{\{i_1, \{i_1, i_2\}\}}}{E}l(1 + \Delta)^{l - 1} + \frac{\beta S_{\{\{i_1, i_2\}, \{i_1, i_3\}\}}}{E^2}l(l-1)(1 + \Delta)^{l - 2} + \frac{\beta S_{\{\{\alpha, i_1\}\}}}{E\delta}l(1 + \Delta)^{l-1}\right) + O(1)\Bigg) \\
& = k\left(\beta - \frac{\alpha}{\delta}\right)h(1 + \Delta) - \frac{2(1 + \Delta)^2\beta h''(1+\Delta)}{\delta^2} + \frac{(1 + \Delta)\beta h'(1+\Delta)}{\delta^2} \\
& \qquad - k\left(\frac{S_{\{i_1, \{i_1, i_2\}\}}}{E}h'(1 + \Delta) + \frac{\beta S_{\{\{i_1, i_2\}, \{i_1, i_3\}\}}}{E^2}h''(1 + \Delta) + \frac{\beta S_{\{\{\alpha, i_1\}\}}}{E\delta}h'(1 + \Delta)\right) + O(1)
\end{align*}
\end{proof}
\begin{corollary}\label{cor:zeroderivatives}
If we have $h(x) = \sum_{l = 1}^ma_lx^l$ such that $h'(1) = h''(1) = 0$ and choose coefficients 
\[
c_{2l+1} = a_l \cdot \frac{2^{k-2}(u-1)!(v-1)!}{(k - 2l - 2)!\delta^{2l}k^{2l-1}E^l}
\]
then the contribution of degree $\geq 3$ terms is
\[
k\left(\beta - \frac{\alpha}{\delta}\right)h(1 + \Delta) + O(k)\cdot \Delta + O(1).
\]
\end{corollary}
\begin{proof}
This follows from the fact that $E = \Omega(k^2), \beta = O(k), S_{\{i_1, \{i_1, i_2\}\}} = O(k^2), S_{\{\{i_1, i_2\}, \{i_1, i_3\}\}} = O(k^3)$ and $S_{\{\{\alpha, i_1\}\}} = O(k)$.
\end{proof}

\subsection{Conditions on $h$}
\begin{lemma}\label{lem:hconditions}
For all $\delta_0 > 0$, if $h = \sum_{l = 1}^{m}{{a_l}x^l}$ is a polynomial such that
\begin{enumerate}
    \item $h'(1) = h''(1) = 0$
    \item For all $\Delta \in [-0.55,\frac{1}{\delta_0^2}]$, $|h(1+\Delta)-1| \leq \frac{{\delta^2_0}|\Delta|}{5}$
    \item For all $\Delta \in [-1,-0.55]$, $0 \leq h(1+\Delta) \leq 1$
\end{enumerate}
then there exists $k_0 \in \mathbb{N}$ such that for all $k \geq k_0$ and all $\delta \geq \delta_0$, the rounding scheme $R_{k,\delta,h}$ has positive expected value over the entire KTW polytope.
\end{lemma}
To prove this, we need the following lemma about points in the KTW polytope for $P$:
\begin{lemma}\label{lem:linear2}
For sufficiently large $k$ we have
\[
\delta k\alpha + \beta \geq \frac{(\delta^2 k - 1)|\Delta|}{4} + \frac{1}{2}.
\]
\end{lemma}
\begin{proof}
Since $|\Delta|$ is a convex function on the KTW polytope, it suffices to check that for each satisfying assignment, $\delta k\alpha + \beta \geq \frac{\delta k|\Delta|}{4} + \frac{1}{2}$. Letting $t$ be the number of ones in $x_2, \ldots, x_k$, we have that $\beta = t - (k - 1 - t) = 2t - k + 1$ and 
\[
\sum_{2 \leq i < j}x_ix_j = \binom{t}{2} + \binom{k - 1 - t}{2} - t(k - 1 - t) = 2t^2 - 2(k - 1)t + \binom{k - 1}{2},
\]
Recalling that $E = \frac{\delta^2k^2}{2} - \frac{k}{2} + 1$, this implies that
\begin{align*}
\Delta = \frac{\sum_{2 \leq i < j}x_ix_j - E}{E} & = \frac{1}{E}\left(2t^2 - 2(k-1)t+\binom{k-1}{2} - \frac{\delta^2k^2}{2} + \frac{k}{2} - 1\right) \\
& = \frac{1}{E}\left(2t^2 - 2(k-1)t + \frac{(1-\delta^2)k^2}{2} - k\right)
\end{align*}
Since $E > \frac{\delta^2k^2}{2} - \frac{k}{2}$, we have
\begin{align*}
    \frac{\delta^2 k - 1}{4}|\Delta| & =  \frac{\delta^2 k - 1}{4E}\left|2t^2 - 2(k-1)t + \frac{(1-\delta^2)k^2}{2} - k\right| \\
    & < \frac{1}{2k}\left|2t^2 - 2(k-1)t + \frac{(1-\delta^2)k^2}{2} - k\right| \\
    & = \left|\frac{t^2}{k} - \frac{(k-1)t}{k} + \frac{(1-\delta^2)k}{4} - \frac{1}{2}\right|
\end{align*}

We will show that $\delta k\alpha + \beta \geq \left|\frac{t^2}{k} - \frac{(k-1)t}{k} + \frac{(1-\delta^2)k}{4} - \frac{1}{2}\right| + \frac{1}{2}$, from which our lemma will follow. To this end, we show that $\delta k\alpha + \beta \geq \left(\frac{t^2}{k} - \frac{(k-1)t}{k} + \frac{(1-\delta^2)k}{4} - \frac{1}{2}\right) + \frac{1}{2}$ and $\delta k\alpha + \beta \geq -\left(\frac{t^2}{k} - \frac{(k-1)t}{k} + \frac{(1-\delta^2)k}{4} - \frac{1}{2}\right) + \frac{1}{2}$.

\begin{itemize}
    \item $\delta k\alpha + \beta \geq \left(\frac{t^2}{k} - \frac{(k-1)t}{k} + \frac{(1-\delta^2)k}{4} - \frac{1}{2}\right) + \frac{1}{2}$. 
    
    We have two cases, $\alpha = 1$ or $\alpha = -1$. If $\alpha = 1$, then $\delta k\alpha + \beta = \delta k + 2t - k + 1$ and since it's a satisfying assignment we have $t \geq \frac{1-\delta}{2}k$. The inequality becomes
    \[
    \frac{t^2}{k} - \frac{(k-1)t}{k} -2t + \frac{(1-\delta^2)k}{4} -\delta k + k - 1 \leq 0.
    \]
    The left hand side is a quadratic function on $t$ with positive leading coefficient, and to check it's non-positive we simply need to check its values on $t = \frac{1-\delta}{2}k$ and $t = k - 1$, the boundary points of $t$'s domain. When $t = \frac{1-\delta}{2}k$,
    \begin{align*}
    &\frac{t^2}{k} - \frac{(k-1)t}{k} -2t + \frac{(1-\delta^2)k}{4} -\delta k + k - 1 = \\
&\left(\frac{t^2}{k} - t + \frac{(1-\delta^2)k}{4}\right) + \left(-2t -\delta k + k\right) + \left(\frac{t}{k} - 1\right) = \\
&\left(\frac{{(1 - \delta)^2}k}{4} - \frac{(1-\delta)k}{2} + \frac{(1-\delta^2)k}{4}\right) + 0 + \left(\frac{1-\delta}{2} - 1\right) = \\
&\frac{1-\delta}{2} - 1 < 0
    \end{align*}
    When $t = k - 1$,
    \begin{align*}
    &\frac{t^2}{k} - \frac{(k-1)t}{k} -2t + \frac{(1-\delta^2)k}{4} -\delta k + k - 1 = \\
    &\left(\frac{t^2}{k} - \frac{(k-1)t}{k}\right) + \left(-t + k - 1\right) + \left(-t - {\delta}k + \frac{(1-\delta^2)k}{4}\right) = \\
    &-\frac{3 - 4\delta + {\delta}^2}{4}k + 1
    \end{align*}
    which is negative when $k$ is sufficiently large (note that $\delta \leq 1 - \frac{3}{k}$).
    
    If $\alpha = -1$, then $\delta k\alpha + \beta = -\delta k + 2t - k + 1$ and we have that $t \geq \frac{1 + \delta}{2}k$. The inequality becomes
    \[
    \frac{t^2}{k} - \frac{(k-1)t}{k} -2t + \frac{(1-\delta^2)k}{4} +\delta k + k - 1 \leq 0.
    \]
    We check the value of LHS on $t = \frac{1 + \delta}{k}$ and $t = k - 1$. Following exactly the same argument we used for $\alpha = 1$ except that $\delta$ is replaced by $-\delta$, when $t = \frac{1 + \delta}{2}k$,
    \[
    \frac{t^2}{k} - \frac{(k-1)t}{k} -2t + \frac{(1-\delta^2)k}{4} +\delta k + k - 1 = 
    \frac{1 + \delta}{2} - 1 < 0
    \]
    and when $t = k - 1$,
    \[
    \frac{t^2}{k} - \frac{(k-1)t}{k} -2t + \frac{(1-\delta^2)k}{4} +\delta k + k - 1 = 
    -\frac{3 - 4\delta + \delta^2}{4}k + 1 < 0
    \]
    \item $\delta k\alpha + \beta \geq -\left(\frac{t^2}{k} - \frac{(k-1)t}{k} + \frac{(1-\delta^2)k}{4} - \frac{1}{2}\right) + \frac{1}{2}$. 
    
    We again have two cases, $\alpha = 1$ or $\alpha = -1$. If $\alpha = 1$, we have $\delta k\alpha + \beta = \delta k + 2t - k + 1$ and the inequality becomes
    \[
    \frac{t^2}{k} + \left(2 - \frac{k - 1}{k}\right)t + \frac{(1-\delta^2)k}{4} + \delta k - k \geq 0.
    \]
    The left hand side is a quadratic function that achieves minimum when $t$ is negative, so we simply need the inequality to hold when $t = \frac{1 - \delta}{2}k$, at which point the value of LHS is $\frac{1 - \delta}{2} \geq 0$.
    
    If $\alpha = -1$, the inequality becomes
    \[
    \frac{t^2}{k} + \left(2 - \frac{k - 1}{k}\right)t + \frac{(1-\delta^2)k}{4} - \delta k - k \geq 0.
    \]
    Again, we simply need it to hold when $t = \frac{1 + \delta}{2}k$, at which point the value of LHS is $\frac{1 + \delta}{2} \geq 0$.
\end{itemize} 
This completes our proof.
\end{proof}
\begin{proof}[Proof of Lemma~\ref{lem:hconditions}]
Since $h'(1) = h''(1) = 0$, by Corollary~\ref{cor:zeroderivatives} the contribution of degree $\geq 3$ terms becomes
\[
k\left(\beta - \frac{\alpha}{\delta}\right)h(1 + \Delta) + O(k)\cdot \Delta + O(1).
\]

Now we add in the contribution of degree 1 terms. Since $\hat{P}_P$ is extremely close to $1$ and $\hat{P}_C$ is exponentially small, the contribution of degree 1 terms is extremely close to ${c_1}\alpha = (\delta k^2 + k/\delta)\alpha$. Adding this to the contribution from the higher degree terms, we get that
\[
\sum_{I \subseteq [k]:I \neq \emptyset}{\hat{P}_{I}\E[x_I]} = k\left(\beta - \frac{\alpha}{\delta}\right)h(1 + \Delta) + \left(\delta k^2 + \frac{k}{\delta}\right)\alpha + O(k)\cdot \Delta + O(1).\tag{$*$}
\]
To establish the theorem, we need to show that $(*)$ is positive over the entire KTW polytope. We proceed with a case analysis on $\Delta$. Note that the range of $\Delta$ is approximately $(-1 - O(1/k), 1/\delta^2)$. We have the following cases.
\begin{itemize}
    \item $\Delta \geq -0.55$. In this case we have
    \begin{align*}
        (*) & = k\left(\beta - \frac{\alpha}{\delta}\right)h(1 + \Delta) + \left(\delta k^2 + \frac{k}{\delta}\right)\alpha + O(k)\cdot \Delta + O(1) \\
        & = k(\delta k \alpha + \beta) + k\left(\beta - \frac{\alpha}{\delta}\right)(h(1 + \Delta) - 1) + O(k)\cdot \Delta + O(1)\\
        & \geq k\left(\frac{(\delta^2 k - 1)|\Delta|}{4} + \frac{1}{2}\right) + k\left(\beta - \frac{\alpha}{\delta}\right)(h(1 + \Delta) - 1) + O(k)\cdot \Delta + O(1)
    \end{align*}
    The last inequality is due to Lemma~\ref{lem:linear2}. Here the two terms which are quadratic in $k$ are ${\delta}^{2}k^2|\Delta|/4$ (note that $\delta > \delta_0$ is at least a constant) and $k\beta(h(1 + \Delta) - 1)$. Since $|\beta| < k$ and $|h(1 + \Delta) - 1| \leq \frac{\delta_0^2|\Delta|}{5}$, the above quantity is positive when $k$ is sufficiently large.
    \item $\Delta < -0.55$. Note that if $x_1 = -1$ then the minimum value of $\Delta$ is about 0. If $x_1 = 1$, then the minimum value of $\Delta$ is about $-1$. This means that when $\Delta < -0.55$, with probability $> 0.5$ we have $x_1 = 1$, which implies $\alpha > 0$ and is $\Omega(1)$. We can write $(*)$ as
    \[
    (*) = k(\delta k\alpha + \beta)h(1 + \Delta) + \left(\delta k^2 + \frac{k}{\delta}\right)\alpha(1 - h(1 + \Delta)) + O(k)\cdot\Delta + O(1).
    \]
    If $\Delta \geq -1$, then $h(1 + \Delta) \in [0,1]$ and both the first two terms are positive and at least one of the two terms is $\Omega(k^2)$. If $\Delta \leq -1$ then since $\Delta \geq -1 - O(1/k)$ we know that the first term is $O(k)$ and the second term is $\Omega(k^2)$ and positive. Either way, we get a positive value when $k$ is sufficiently large. 
\end{itemize}
\end{proof}

\subsection{Choosing the Rounding Polynomial}

In this subsection, we construct a polynomial that satisfies the conditions in Lemma~\ref{lem:hconditions}. We will first show that $h(x) = 1 - (1 - x)^3\exp(-Bx)$ works for some constant $B$ except that it's not a polynomial. We then show that by truncating the Taylor expansion of this function we can get a polynomial which also works.

\begin{lemma}\label{lem:before_truncate}
Let $h(x) = h_1(x) = 1 - (1 - x)^3\exp(-Bx)$, where $B = \max\left(\frac{5}{\delta_0}, \frac{1}{0.45}\ln\frac{5}{\delta_0^2}\right)$. Then $h$ satisfies the conditions in Lemma~\ref{lem:hconditions} except that it is not a polynomial.
\end{lemma}

Substituting $x$ by $1 + \Delta$, we have $h(1 + \Delta) = 1 + \Delta^3\exp(-B(1 + \Delta))$. Below is a plot of this function.

\begin{figure}[H]
\centering
\begin{tikzpicture}[scale = 2]
      \draw[->] (-0.1,0) -- (2.5,0) node[right] {$\Delta$};
      \draw[->] (0,-0.1) -- (0,1.5) node[above] {multiple of $\beta$};
      \draw[domain=-0.1:2.5,smooth,variable=\x,blue] plot ({\x},{1});
      \draw[domain=-0.02:2.5,smooth,variable=\x,red] plot ({\x},{1 - (1 - \x)^3 * (2.71828)^(-4 * \x)});
      \draw[dotted] (1, -0.1) -- (1, 2);
      \node[red] at (1.5, 1.3) {$h = 1 + \Delta^3\exp(-B(1 + \Delta))$};
      \node at (-0.3, 1) {$1$};
      \node at (-0.3, 0) {$0$};
      \node at (1, -0.3) {$0$};
      \node at (0, -0.3) {-$1$};
\end{tikzpicture}
\caption{Plot of $h = 1 + \Delta^3\exp(-B(1 + \Delta))$.}
\end{figure}
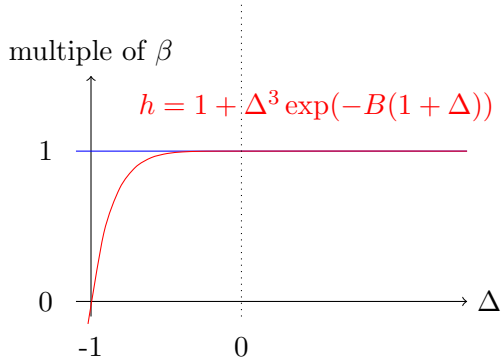

\begin{proof}[Proof of Lemma~\ref{lem:before_truncate}]
Item 1 in Lemma~\ref{lem:hconditions} is clearly satisfied. We now prove item 2. We need to show that
    \[
    \frac{\delta_0^2|\Delta|}{5} \geq |h(1 + \Delta) - 1| = |\Delta^3\exp(-B(1 + \Delta))|,
    \]
    that is,
    \[
    \Delta^2\exp(-B(1 + \Delta)) \leq \frac{\delta_0^2}{5}.
    \]
    Let $g(\Delta) = \Delta^2\exp(-B(1 + \Delta))$, we have $g'(\Delta) = \Delta(2 - B\Delta)\exp(-B(1 + \Delta))$. Since $\Delta \geq -0.55$, the maximum of $g$ is either $g(-0.55)$ or $g(2/B)$. Since $B = \max\left(\frac{5}{\delta_0}, \frac{1}{0.45}\ln\frac{5}{\delta_0^2}\right)$, we have 
    \[
    g(2/B) = \frac{4}{B^2}\exp(-B - 2)< \frac{4}{B^2} \leq \frac{4}{(5/\delta_0)^2} < \frac{\delta_0^2}{5},
    \]
    and 
    \[
    g(-0.55) = 0.55^2\exp(-0.45B) \leq 0.55^2 \cdot \frac{\delta_0^2}{5} < \frac{\delta_0^2}{5}.
    \]
Note that we actually obtained strict inequality in both cases. For item 3, notice that $h(1 + \Delta)$ is monotone for $\Delta \in [-1, 0]$ and $h(0) = 0, h(1) = 1$.
\end{proof}
We now truncate the Taylor expansion of the function in the above lemma.
\begin{lemma}\label{lem:after_truncate}
There exists $m \in \mathbb{N}$ such that, $h(x) = h_2(x) = 1 - (1 - x)^3\sum_{l = 0}^m\frac{(-Bx)^l}{l!}$ satisfies the conditions in Lemma~\ref{lem:hconditions} where $B = \max\left(\frac{5}{\delta}, \frac{1}{0.45}\ln\frac{5}{\delta^2}\right)$.
\end{lemma}
\begin{proof}
First of all, this polynomial has no constant term as $h_2(0) = 0$. It is also straightforward to see that $h'_2(1) = h''_2(1) = 0$. Since the Taylor expansion of the exponential function is uniformly convergent\footnote{For the notion of uniform convergence, see for example in~\cite{rudin_principles_1976}. The uniform convergence of the Taylor expansion of $\exp(x)$ can be easily obtained by Weierstrass test.} on any bounded interval, for any $\eta > 0$ we can choose $m \in \mathbb{N}$ such that for every $\Delta \in [-2, 1/\delta^2]$.
\[
\left|\exp(-B(1 + \Delta)) - \sum_{l = 0}^m\frac{(-B(1 +\Delta))^l}{l!}\right| \leq \eta.
\]
We now verify the second and third items in Lemma~\ref{lem:hconditions}. For item 2, We need to show that
    \[
    \frac{\delta_0^2|\Delta|}{5} \geq |h(1 + \Delta) - 1| = \left|\Delta^3\left(\sum_{l = 0}^m\frac{(-Bx)^l}{l!}\right)\right|,
    \]
    that is,
    \[
    \Delta^2\left|\sum_{l = 0}^m\frac{(-Bx)^l}{l!}\right| \leq \frac{\delta_0^2}{5}.
    \]
    In the proof of previous lemma we showed that $\Delta^2\exp(-B(1 + \Delta)) < \frac{\delta_0^2}{5}$, which means that we can choose $\eta$ and $m$ appropriately so that the above inequality is still satisfied. For item 3, notice that $h_1(0) = h_2(0) = 0$ and as long as $m \geq 1$, $h_1'(0) = h_2'(0)$, so $h_2$ will be in $[0, 1]$ on $(0,\epsilon]$ for some small $\epsilon$. Then, if we choose $\eta < \min(h_1(\epsilon), 1 - h_1(1-0.55))$, we can make sure that $h_2$ is in $[0, 1]$ on $[\epsilon, 1-0.55]$ as well.
\end{proof}

By Lemma~\ref{lem:after_truncate}, if we choose $h_2$ as the rounding polynomial, our rounding scheme will have a positive expected value for any point in the KTW polytope. This completes the proof of Theorem~\ref{thm:3}, which then implies Theorem~\ref{thm:2}.

\section{Evidence for the Necessity of Pairwise Biases}
We have now given rounding schemes for almost all presidential predicates. These rounding schemes crucially use the pairwise biases $\{b_{ij}: i < j \in [k]\}$. A natural question is whether this is necessary or it is possible to only use the biases $\{b_i: i \in [k]\}$. If there is a rounding scheme which only uses the biases, then instead of using a semidefinite program, it is sufficient to use a linear program, which is much faster. Indeed, such rounding schemes exist for predicates which are close to the majority function \cite{hast_beating_2005} and for the monarchy predicate \cite{austrin_approximation_2009, potechin_approximation_2018}. 

In this section, we give evidence that this is not possible for more general presidential type predicates and it is necessary to use the pairwise biases. In particular, we prove the following theorem.


Recall that we choose a rounding scheme by specifying $f_a(b_{i_1}, \ldots, b_{i_a}, b_{i_1i_2}, \ldots, b_{i_{a-1}i_a})$ for each $a \in [k]$.
\begin{definition}
We say that a rounding scheme has degree $m$ if $f_{m} \neq 0$ and $f_{a} = 0$ for all $a > m$.
\end{definition}
\begin{definition}
We say that a rounding scheme does not use pairwise biases if for all $a \in [k]$, $f_a(b_{i_1}, \ldots, b_{i_a}, b_{i_1i_2}, \ldots, b_{i_{a-1}i_a})$ only depends on $\{b_{i_1}, \ldots, b_{i_a}\}$.
\end{definition}
\begin{theorem}\label{thm:pairwisenecessity}
For all $\delta_0 > 0$ and all $m \in \mathbb{N}$, there exists a $k_0$ such that for all $k \geq k_0$ and $\delta \in (\delta_0,1 - 4/k]$ where ${\delta}k + k - 1$ is an odd integer, the presidential type predicate $P(x) = \sign\left({\delta}k{x_1} + \sum_{i=2}^{k}{x_i}\right)$ cannot be approximated by any rounding scheme of degree at most $m$ which does not use pairwise biases.
\end{theorem}
\begin{proof}
Let us consider a two-player zero-sum game where Alice chooses a point $b \in [-1, 1]^k$ in the KTW polytope\footnote{The KTW polytope of $P$ actually has dimension $k + \binom{k}{2}$, but since pairwise biases play no role here, we omit those coordinates for simplicity.} of $P$ and Bob chooses a rounding scheme $R$ of degree at most $m$. The objective of Alice is to minimize $R(b) = \sum_{I \subseteq [k]:I \neq \emptyset}{\hat{P}_{I}\E[x_I]}$, the expected value of $P(x)$ if we are given the point $b$ in the KTW polytope and apply the rounding scheme $R$. 

The lemma will follow if we can show a mixed strategy for Alice, which is a distribution $\mu$ over points in KTW polytope, such that for any rounding scheme $R$, $\E_{b \sim \mu}[R(b)] = 0$. Recalling that for each $a \in [m]$ and monomial $x_{i_1}x_{i_2}\ldots x_{i_a}$ of degree $a$, $\E[x_{i_1}x_{i_2}\ldots x_{i_a}] = f_a(b_{i_1}, \ldots, b_{i_a})$, it suffices to have the sum of degree $a$ terms be zero for every $a \in [m]$, i.e.,
\[
\E_{b \sim \mu}\left[\sum_{\substack{I \subset [k], |I| = a\\I = \{i_1, \ldots, i_a\}}}{\hat{P}_I}f_a(b_{i_1}, \ldots, b_{i_a})\right] = 0, \quad \forall a \in [m].
\]

Now let us construct such a distribution $\mu$. By Lemma~\ref{lem:fourier_coeff}, $\hat{P}_P$, the Fourier coefficient of the president $x_1$, is exponentially larger than $\hat{P}_C$, and $\lim_{k \to \infty}\hat{P}_{P+(t-1)C} / \hat{P}_{tC} = -1$ for every odd integer $t \leq m$. 
For concreteness, let us assume that $m = 5$. Then we will have the following distribution for $\mu$:
\begin{center}
    \begin{tabular}{c|c|cccccccc}
        \text{Probability} & $x_1$ & $x_2$ & $x_3$ & $x_4$ & $x_5$  & $x_6$ & $x_7$ & $\cdots$ & $x_k$\\ \hline
        $p_1$ & 0 & 1 & 0 & 0 & 0 & 0 & 0 & $\cdots$ & 0 \\\hline
        $p_2$ & 0 & 1 & 1 & -1 & 0 & 0 & 0 & $\cdots$ & 0 \\\hline
        $p_3$ & 0 & 1 & 1 & 1 & 0 & 0 & 0 & $\cdots$ & 0 \\\hline
        $p_4$ & 0 & 1 & 1 & 1 & 1 & -1 & 0 & $\cdots$ & 0 \\\hline
        $p_5$ & -1 & 1 & 1 & 1 & 1 & 1 & 1 & $\cdots$ & 1 \\
    \end{tabular}   
\end{center}
First of all, it is easy to check that all these points are inside the KTW polytope for $P(x)$. The following is a table of contribution of each degree from each of these points.
\begin{center}
    \begin{tabular}{|c|c|c|c|c|c|}
    \hline \backslashbox{Degrees}{Points} & \text{1st type} & \text{2nd type}  & \text{3rd type}  & \text{4th type}  & \text{5th type} \\ \hline
    $f_1(1)$ & $p_1 \cdot \hat{P}_C$ & $p_2 \cdot \hat{P}_C$ & $p_3 \cdot 3\hat{P}_C$ & $p_4 \cdot 3\hat{P}_C$ & $p_5\cdot(-\hat{P}_P + (k-1)\hat{P}_C)$ \\ \hline
    $f_3(1,1,1)$ & 0 & $-p_2 \cdot \hat{P}_{3C}$ & $p_3 \cdot \hat{P}_{3C}$ & $-2p_4 \cdot \hat{P}_{3C}$ & $p_5\cdot(\binom{k-1}{2}\hat{P}_{P+2C} + \binom{k-1}{3}\hat{P}_{3C})$ \\ \hline
    $f_5(1,1,1,1,1)$ & 0 & 0 & 0 & $-p_4 \cdot \hat{P}_{5C}$ & $p_5\cdot(\binom{k-1}{4}\hat{P}_{P+4C} + \binom{k-1}{5}\hat{P}_{5C})$ \\ \hline
    
    \end{tabular}   
\end{center}
To balance degree 1 terms, we need
\[
p_1 \cdot \hat{P}_C + p_2 \cdot \hat{P}_C + p_3 \cdot 3\hat{P}_C + p_4 \cdot 3\hat{P}_C + p_5\cdot(-\hat{P}_P + (k-1)\hat{P}_C) = 0.
\]
Notice that every point in this distribution has a positive contribution from citizens (i.e., variables $x_2, \ldots, x_k$), so we need a negative contribution from $x_1$. Since $\hat{P}_P$ is exponentially larger than $\hat{P}_C$, we can achieve the balance by having $p_5$ be exponentially small in $k$. Then we balance degree 5 terms, for which we need
\[
-p_4 \cdot \hat{P}_{5C} + p_5\cdot\left(\binom{k-1}{4}\hat{P}_{P+4C} + \binom{k-1}{5}\hat{P}_{5C}\right) = 0.
\]
Recall that $\lim_{k \to \infty}\hat{P}_{P+4C} / \hat{P}_{5C} = -1$, so we can achieve the balance by having $p_4 = poly(k)\cdot p_5$, where $poly(k)$ is a polynomial in $k$. For degree 3 terms, we need
\[
-p_2 \cdot \hat{P}_{3C} + p_3 \cdot \hat{P}_{3C} - 2p_4 \cdot \hat{P}_{3C} + p_5\cdot\left(\binom{k-1}{2}\hat{P}_{P+2C} + \binom{k-1}{3}\hat{P}_{3C}\right) = 0.
\]
We can then use either the second type or the third type to balance degree 3 terms. Again we will only use $poly(k)\cdot p_5$ amount of probability. When $k$ is sufficiently large, $p_2 + p_3 + p_4 + p_5 \leq 1$ and we let the first type of points take up the remaining probability. This method can be easily extended to handle the case where $m$ is any fixed positive integer. 
\end{proof}
\begin{remark}
This argument fails for the monarchy predicate $P(x) = sign\left((k-2)x_1 + \sum_{i=2}^{k}{x_i}\right)$ for the following reason. The only satisfying assignment to the monarchy predicate where $x_1 = -1$ is when all of the other $x_i$ are $1$. This implies that for all $i \in [2,k]$, $b_i \geq -b_1$, which means that the point $b = (0,1,1,-1,0,\ldots,0)$ and similar points are not in the KTW polytope.
\end{remark}
\begin{remark}
This theorem rules out any fixed degree rounding schemes that use only biases, but it does not rule out the possibility that a rounding scheme might be able to succeed with just biases if its degree grows with $k$.
\end{remark}

\section{Conclusions}
In this paper, we showed that almost all presidential type predicates are approximable. To do this, we carefully constructed rounding schemes which have positive expected value over the entire KTW polytope. These rounding schemes use both the biases $\{b_i:i \in [k]\}$ and the pairwise biases $\{b_{ij}: i < j \in [k]\}$ and have relatively high (but still constant) degree.

This work raises a number of open questions, including the following:
\begin{enumerate}
    \item Which other types of predicates can this technique be applied to? For example, can we show that almost all oligarchy-type predicates are approximable, where oligarchy-type predicates are balanced LTFs where all but a few of the inputs have the same weight? 
    
    As another example, can we extend the result of Austrin, Bennabas, and Magen that all symmetric quadratic threshold functions with no constant term are approximable to show that almost all quadratic threshold functions with no constant term which are symmetric with respect to all but one variable are approximable or at least weakly approximable?
    \item Can we show that for almost all presidential type predicates, there is no rounding scheme which only uses the biases $\{b_i: i \in [k]\}$? Note that by Theorem ~\ref{thm:pairwisenecessity}, such rounding schemes would have to have degree which increases with $k$.
    \item Our results only hold if $k$ is sufficiently large. Is it true that all presidential type predicates are approximable? Less ambitiously, can we either extend our techniques or develop new techniques to handle presidential type predicates where $k$ is relatively small?
\end{enumerate}

\bibliographystyle{plain}
\bibliography{references}

\appendix
\section{Missing Proofs from Section 2}
{
\renewcommand{\thetheorem}{\ref{lem:fourier_coeff}}
\begin{lemma}[Restated]
Let $P(x_1, \ldots, x_k) = \sign(a\cdot x_1 + x_2 + \cdots + x_k)$ be a presidential type predicate where $a \leq k - 2$ and $a + k - 1$ is an odd integer. Let $\hat{P}_{tC}$ denote the Fourier coefficient of a set of $t$ citizens (indices from $2$ to $k$) and $\hat{P}_{P+tC}$ denote the Fourier coefficient of a set of $t$ citizens together with the president (index 1). Let $\tau = \lfloor (k - a - 1)/2 \rfloor$. We have
\begin{enumerate}[(1)]
    \item $\hat{P}_P = 1 - \frac{1}{2^{k - 2}}\sum_{l = 0}^\tau\binom{k-1}{l},$
    \item $\hat{P}_{tC} = \frac{1}{2^{k - 2}}\sum_{i = 0}^\tau\sum_{j = 0}^{\tau - i}(-1)^j\binom{k-t-1}{i}\binom{t}{j}, \qquad  \forall t (1 \leq t \leq k - 1 \wedge t \textrm{ is odd}),$
    \item $\hat{P}_{P+tC} = -\frac{1}{2^{k - 2}}\sum_{i = 0}^\tau\sum_{j = 0}^{\tau - i}(-1)^j\binom{k-t-1}{i}\binom{t}{j}, \qquad  \forall t (2 \leq t \leq k - 1 \wedge t \textrm{ is even}).$
\end{enumerate}
\end{lemma}
\addtocounter{theorem}{-1}
}
\begin{proof}[Proof of Lemma~\ref{lem:fourier_coeff}]
$ $
\begin{enumerate}[(1)]

    \item We have 
    \[
    \hat{P}_P = \E_{x \in \{-1, 1\}^k}P(x)x_1 = \frac{1}{2^k}\sum_{x \in \{-1, 1\}^k}P(x)x_1.
    \]
    We first choose how citizens vote. Note that if the vote is already determined by the citizens, then the contribution to the sum is 0. Suppose that at most $\tau = \lfloor (k - a - 1)/2 \rfloor$ citizens vote 1. Then, even if the president also votes 1, we have
    \[
    a\cdot x_1 + x_2 + \cdots + x_k = a + \tau - (k - 1 - \tau) = 2\tau - (k - a - 1) < 0.
    \]
    So no matter how the president votes, we always have $P(x) = -1$. Similarly, if at most $\tau$ citizens vote $-1$, then no matter how the president votes we always have $P(x) = 1$. These two cases contribute 0 to the sum. In the remaining scenarios, the vote of president determines the result, i.e., $P(x) = x_1$. This case contributes $1 - 2 \cdot\frac{1}{2^{k - 1}}\sum_{l = 0}^\tau\binom{k-1}{l} = 1 - \frac{1}{2^{k - 2}}\sum_{l = 0}^\tau\binom{k-1}{l}$.
    \item Let $I$ be a set of $t$ citizens where $t$ is an odd integer. By symmetry we have
    \[
    \hat{P}_{tC} = \E_{x \in \{-1, 1\}^k}P(x)x_I = \frac{1}{2^k}\sum_{x \in \{-1, 1\}^k}P(x)x_I.
    \]
    We analyze the sum as follows. 
    \begin{itemize}
        \item $x_1 = 1$. Assume that $i$ citizens from $\{2, 3, \ldots, k\} - I$ vote 1. If $i > \tau$, then the result is 1 no matter how people in $I$ vote, which means the contribution is 0. Now assume $i \leq \tau$. Let $j$ be the number of citizens from $I$ that vote 1. Note that $P(x) = 1$ if and only if $i + j > \tau$,
        so we have $P(x)x_I = (-1)^{t - j + 1} = (-1)^j$ if $j \leq \tau - i$ and $P(x)x_I = (-1)^{t - j} = (-1)^{j+1}$ if $j > \tau - i$. The contribution in this case is
        \[
        \sum_{i = 0}^\tau \binom{k - 1 - t}{i}\left(\sum_{j = 0}^{\tau - i}\binom{t}{j}(-1)^j + \sum_{j = \tau - i + 1}^t\binom{t}{j}(-1)^{j + 1}\right) = 2\sum_{i = 0}^\tau \binom{k - 1 - t}{i}\sum_{j = 0}^{\tau - i}\binom{t}{j}(-1)^j.
        \]
        The equality comes from the fact that $\sum_{j = 0}^t\binom{t}{j}(-1)^j = (1 - 1)^t = 0$.
        \item $x_1 = -1$. This case is symmetric. Note that since $P$ and $x_I$ are both odd, we have $P(x)x_I = P(-x)(-x)_I$. So the contribution of this case is also $2\sum_{i = 0}^\tau \binom{k - 1 - t}{i}\sum_{j = 0}^{\tau - i}\binom{t}{j}(-1)^j$.
    \end{itemize}
    Summing these contributions up, we obtain $\hat{P}_{tC} = \frac{1}{2^{k - 2}}\sum_{i = 0}^\tau\sum_{j = 0}^{\tau - i}(-1)^j\binom{k-t-1}{i}\binom{t}{j}$.
    \item The analysis in this case is almost similar to that in item (2). Let $I$ be a set of the president along with $t$ citizens where $t$ is an even integer. We have
    \[
    \hat{P}_{P + tC} = \E_{x \in \{-1, 1\}^k}P(x)x_I = \frac{1}{2^k}\sum_{x \in \{-1, 1\}^k}P(x)x_I.
    \]
    We analyze the sum as follows. 
    \begin{itemize}
        \item $x_1 = 1$. Again assume that $i$ citizens from $\{2, 3, \ldots, k\} - I$ vote 1. If $i > \tau$, then the contribution is 0. Now assume $i \leq \tau$. Let $j$ be the number of citizens from $I$ that vote 1. This time we have $P(x)x_I = (-1)^{t - j + 1} = (-1)^{j+1}$ if $j \leq \tau - i$ and $P(x)x_I = (-1)^{t - j} = (-1)^{j}$ if $j > \tau - i$. The contribution in this case is
        \[
        \sum_{i = 0}^\tau \binom{k - 1 - t}{i}\left(\sum_{j = 0}^{\tau - i}\binom{t}{j}(-1)^{j+1} + \sum_{j = \tau - i + 1}^t\binom{t}{j}(-1)^{j}\right) = -2\sum_{i = 0}^\tau \binom{k - 1 - t}{i}\sum_{j = 0}^{\tau - i}\binom{t}{j}(-1)^j.
        \]
        \item $x_1 = -1$. Similarly, the contribution of this case is also $-2\sum_{i = 0}^\tau \binom{k - 1 - t}{i}\sum_{j = 0}^{\tau - i}\binom{t}{j}(-1)^j$.
    \end{itemize}
    Summing these contributions up, we have $\hat{P}_{P+tC} = -\frac{1}{2^{k - 2}}\sum_{i = 0}^\tau\sum_{j = 0}^{\tau - i}(-1)^j\binom{k-t-1}{i}\binom{t}{j}$.
\end{enumerate}
\end{proof}

{
\renewcommand{\thetheorem}{\ref{prop:bias_sum}}
\begin{prop}[Restated] 
For every $l \geq 1$,
\begin{align*}
    \frac{l!}{E^l}S_{1, l} &= \beta(1 + \Delta)^l - \frac{S_{\{i_1, \{i_1, i_2\}\}}}{E}l(1 + \Delta)^{l - 1} - \frac{\beta S_{\{\{i_1, i_2\}, \{i_1, i_3\}\}}}{E^2}l(l-1)(1 + \Delta)^{l - 2} + O\left(\frac{1}{k}\right),\\
    \frac{l!}{E^l}S_{2, l} &= \alpha(1 + \Delta)^l + O\left(\frac{1}{k}\right),\\
    \frac{l!}{E^l}S_{3, l} &= \frac{\beta S_{\{\{\alpha, i_1\}\}}}{E}l(1 + \Delta)^{l-1} + O\left(\frac{1}{k}\right),\\
\end{align*}
where the hidden constants in big-O may depend on $l$.
\end{prop}
\addtocounter{theorem}{-1}
}

\begin{proof}[Proof of Proposition~\ref{prop:bias_sum}]
Here we only prove the first equality since the other two can be proved similarly. Recall that $S_{\{\{i_1, i_2\}\}} = E(1 + \Delta)$ and $E = \Theta(k^2)$. The first equality is equivalent to
\[
l!S_{1, l} = \beta(S_{\{\{i_1, i_2\}\}})^l - S_{\{i_1, \{i_1, i_2\}\}}l(S_{\{\{i_1, i_2\}\}})^{l - 1} - \beta S_{\{\{i_1, i_2\}, \{i_1, i_3\}\}}l(l-1)(S_{\{\{i_1, i_2\}\}})^{l - 2} + O\left(k^{2l-1}\right). 
\]
Let's analyze the term $\beta(S_{\{\{i_1, i_2\}\}})^l$, by definition, it's equal to
\[
\left(\sum_{i \geq 2}b_i\right)\left(\sum_{2 \leq i < j}b_{ij}\right)^l = \sum_{\substack{j_1, j_2, \ldots, j_{2l+1} \in \{2, 3, \ldots, k\} \\ j_2 < j_3, j_4 < j_5, \cdots ,j_{2l} < j_{2l + 1}}}b_{j_1}b_{j_2j_3}b_{j_4j_5}\cdots b_{j_{2l}j_{2l+1}}. 
\]
Let's call the sum on the right hand side $T$. We classify the terms in $T$ according to number of repetitions in indices. If there is no repetition, then the term $b_{j_1}b_{j_2j_3}b_{j_4j_5}\cdots b_{j_{2l}j_{2l+1}}$ is also in $S_{1, l}$. Note that in $S_{1, l}$ the order of the $l$ pairwise biases can be arbitrary, so the sum of terms with no repeated indices is equal to $l!S_{1, l}$. If there are two or more repetitions, then the number of distinct indices is at most $2l-1$, and the contribution of such terms is $O(k^{2l-1})$. If there is exact one repetition, then there are two cases.
\begin{itemize}
    \item $j_1$ is equal to some $j_t$ for $t \geq 2$. Without loss of generality consider the terms where the only repetition is $j_1 = j_2$ or $j_1 = j_3$ (note that $j_2 < j_3$). The contribution of these terms are
    \begin{align*}
    \sum_{\substack{j_1, j_2, \ldots, j_{2l+1} \in \{2, 3, \ldots, k\} \\ j_2 < j_3, j_4 < j_5, \cdots ,j_{2l} < j_{2l + 1} \\ j_1 = j_2 \text{ or } j_1 = j_3 \\
    j_2, j_3, \ldots, j_{2l + 1} \text{ distinct}}}b_{j_1}b_{j_2j_3}b_{j_4j_5}\cdots b_{j_{2l}j_{2l+1}} & = 
    \sum_{\substack{j_1, j_2, \ldots, j_{2l+1} \in \{2, 3, \ldots, k\} \\ j_2 < j_3, j_4 < j_5, \cdots ,j_{2l} < j_{2l + 1} \\ j_1 = j_2 \text{ or } j_1 = j_3}}b_{j_1}b_{j_2j_3}b_{j_4j_5}\cdots b_{j_{2l}j_{2l+1}} + O(k^{2l-1}) \\
    & = S_{\{i_1, \{i_1, i_2\}\}}(S_{\{\{i_1, i_2\}\}})^{l - 1} + O(k^{2l-1}).
    \end{align*}
    This is because the terms where $j_2, j_3, \ldots, j_{2l + 1}$ are not distinct have at most $2l - 1$ distinct indices and contribute $O(k^{2l-1})$. So the contribution of this case is
    \[
    lS_{\{i_1, \{i_1, i_2\}\}}(S_{\{\{i_1, i_2\}\}})^{l - 1} + O(k^{2l-1}).
    \]
    \item $j_s = j_t$ for some $s, t \geq 2$. Note that in this case $s$ and $t$ cannot appear in the same pairwise bias. Without loss of generality assume $s \in \{2, 3\}$ and $t \in \{4, 5\}$. We have
    \begin{align*}
    \sum_{\substack{j_1, j_2, \ldots, j_{2l+1} \in \{2, 3, \ldots, k\} \\ j_2 < j_3, j_4 < j_5, \cdots ,j_{2l} < j_{2l + 1} \\ \text{one repetition in } j_2, j_3, j_4, j_5 \\
    \text{other indices distinct}}}b_{j_1}b_{j_2j_3}b_{j_4j_5}\cdots b_{j_{2l}j_{2l+1}} & = 
    \sum_{\substack{j_1, j_2, \ldots, j_{2l+1} \in \{2, 3, \ldots, k\} \\ j_2 < j_3, j_4 < j_5, \cdots ,j_{2l} < j_{2l + 1} \\ \text{one repetition in } j_2, j_3, j_4, j_5}}b_{j_1}b_{j_2j_3}b_{j_4j_5}\cdots b_{j_{2l}j_{2l+1}} + O(k^{2l-1}) \\
    & = 2\beta S_{\{\{i_1, i_2\}, \{i_1, i_3\}\}}(S_{\{\{i_1, i_2\}\}})^{l - 2} + O(k^{2l-1}).
    \end{align*}
    So the contribution of this case is
    \begin{align*}
     & \binom{l}{2} \cdot \left(2\beta S_{\{\{i_1, i_2\}, \{i_1, i_3\}\}}(S_{\{\{i_1, i_2\}\}})^{l - 2} + O(k^{2l-1})\right) \\
    = \,\,  & \beta S_{\{\{i_1, i_2\}, \{i_1, i_3\}\}}l(l-1)(S_{\{\{i_1, i_2\}\}})^{l - 2} + O\left(k^{2l-1}\right).
    \end{align*}
\end{itemize}
We conclude that 
\[
\beta(S_{\{\{i_1, i_2\}\}})^l = l!S_{1, l} + lS_{\{i_1, \{i_1, i_2\}\}}(S_{\{\{i_1, i_2\}\}})^{l - 1} + \beta S_{\{\{i_1, i_2\}, \{i_1, i_3\}\}}l(l-1)(S_{\{\{i_1, i_2\}\}})^{l - 2} + O\left(k^{2l-1}\right).
\]
We get the desired equality by shifting the terms.
\end{proof}

\section{Proof of Lemma~\ref{lem:fourier_coeff_2}}

We first prove some combinatorial identities to be used later.
\begin{prop}\label{prop:b1}
For $t, l \in \mathbb{N}$, we have
\[
\sum_{j = 0}^l (-1)^j\binom{t}{j} = (-1)^l \cdot \binom{t-1}{l}.
\]
\end{prop}
\begin{proof}
We prove by induction on $l$. If $l = 0$, then $LHS = 1 = RHS$. For $l \geq 1$, we have
\begin{align*}
\sum_{j = 0}^l (-1)^j\binom{t}{j} & = (-1)^{l-1} \cdot \binom{t-1}{l-1} + (-1)^l\binom{t}{l} \\
& = (-1)^{l-1} \cdot \binom{t-1}{l-1} + (-1)^l\left(\binom{t-1}{l-1} + \binom{t-1}{l}\right) \\
& = (-1)^l \cdot \binom{t-1}{l}.
\end{align*}
\end{proof}

\begin{prop}\label{prop:b2}
For $a, b \in \mathbb{R}$, $k \in \mathbb{N}$, we have
\begin{enumerate}
    \item $\sum_{i = 0}^k\binom{k}{i}a^{i-1}b^{k - i}i(i + 1) = 2k(a+b)^{k-1} + k(k-1)a(a+b)^{k-2}$.
    \item $\sum_{i = 0}^k\binom{k}{i}a^{i}b^{k-i-1}(k-i)(k-i+1) = 2k(a+b)^{k-1} + k(k-1)b(a+b)^{k-2}$.
\end{enumerate}
\end{prop}
\begin{proof}
We have
\begin{align*}
    \sum_{i = 0}^k\binom{k}{i}a^{i-1}b^{k - i}i(i + 1) & = \frac{\partial^2}{\partial^2 a}\left(\sum_{i = 0}^k\binom{k}{i}a^{i+1}b^{k - i}\right) \\
    & = \frac{\partial^2}{\partial^2 a}\left(a(a+b)^k\right) \\
    & = 2k(a+b)^{k-1} + k(k-1)a(a+b)^{k-2}.
\end{align*}
This gives Item 1. By substituting $i$ with $k-i$ and swapping $a$ and $b$ in Item 1 we get Item 2.
\end{proof}

Now we are ready to prove Lemma~\ref{lem:fourier_coeff_2}.
{
\renewcommand{\thetheorem}{\ref{lem:fourier_coeff_2}}
\begin{lemma}[Restated]
Let $P(x_1, \ldots, x_k) = \sign(\delta\cdot k x_1 + x_2 + \cdots + x_k)$ where $\delta \in (0, 1)$ such that $\delta k + k - 1$ is an odd integer. Let $u = \frac{1+\delta}{2}k$ and $v = \frac{1-\delta}{2}k$. Let $\hat{P}_{tC}$ denote the Fourier coefficient of a set of $t$ citizens and $\hat{P}_{P+tC}$ denote the Fourier coefficient of a set of $t$ citizens together with the president. We have
\begin{align*}
    &\hat{P}_P = 1 - \frac{1}{2^{k - 2}}\sum_{l = 0}^{v - 1}\binom{k-1}{l}, \\
&\hat{P}_{tC} = \frac{1}{2^{k - 2}}\cdot\frac{(k - t - 1)!}{(u-1)!(v-1)!}\left(\delta^{t-1}k^{t-1} - \frac{(t-1)(t-2)}{2}\delta^{t-3}k^{t - 2} + O(k^{t - 3})\right), \quad t \textrm{ is an odd constant}\\
    &\hat{P}_{P+tC} = -\frac{1}{2^{k - 2}}\cdot\frac{(k - t - 1)!}{(u-1)!(v-1)!}\left(\delta^{t-1}k^{t-1} - \frac{(t-1)(t-2)}{2}\delta^{t-3}k^{t - 2} + O(k^{t - 3})\right), \quad t \textrm{ is an even constant}\\
\end{align*}
where the constants inside the big $O$s depend on $t$ but not on $\delta$.
\end{lemma}
\addtocounter{theorem}{-1}
}
\begin{proof}[Proof of Lemma~\ref{lem:fourier_coeff_2}]
We have $\tau = \lfloor (k - \delta k - 1)/2 \rfloor = v - 1$. It follows from Lemma~\ref{lem:fourier_coeff} that 
\begin{align*}
    &\hat{P}_P = 1 - \frac{1}{2^{k - 2}}\sum_{l = 0}^{v - 1}\binom{k-1}{l}, \\
&\hat{P}_{tC} = \frac{1}{2^{k - 2}}\sum_{i = 0}^{v-1}\sum_{j = 0}^{v-1 - i}(-1)^j\binom{k-t-1}{i}\binom{t}{j}, \qquad  \forall t (1 \leq t \leq k - 1 \wedge t \textrm{ is odd}),\\
    &\hat{P}_{P+tC} = -\frac{1}{2^{k - 2}}\sum_{i = 0}^{v-1}\sum_{j = 0}^{v-1 - i}(-1)^j\binom{k-t-1}{i}\binom{t}{j}, \qquad  \forall t (2 \leq t \leq k - 1 \wedge t \textrm{ is even}).\\
\end{align*}
We have
\begin{align*}
\sum_{i = 0}^{v-1}\sum_{j = 0}^{v-1 - i}(-1)^j\binom{k-t-1}{i}\binom{t}{j} & = \sum_{i = 0}^{v-1}\binom{k-t-1}{i}\left(\sum_{j = 0}^{v-1 - i}(-1)^j\binom{t}{j}\right) \\
& = \sum_{i = 0}^{v-1}\binom{k-t-1}{i}(-1)^{v-1-i}\binom{t-1}{v-1-i}\\
& = \sum_{i = v-t}^{v-1}\binom{k-t-1}{i}(-1)^{v-1-i}\binom{t-1}{v-1-i}. 
\end{align*}
The second to last equaltiy comes from Proposition~\ref{prop:b1}, and the last equality comes from the fact that $\binom{t-1}{v-1-i} = 0$ if $t-1 < v-1-i$. For $v-t\leq i \leq v-1$ we have
\begin{align*}
\binom{k-t-1}{i} & = \frac{(k-t-1)!}{i!(k-t-1-i)!} \\ 
& =\frac{(k-t-1)!}{(u-1)!(v-1)!}\cdot(i+1)(i+2)\cdots(v-1)(k-t-i)(k-t-i+1)\cdots(u-1) \\
& = \frac{(k-t-1)!}{(u-1)!(v-1)!}\prod_{l = 1}^{v - 1 - i}(v-l)\prod_{l = 1}^{t-v+i}(u-l).
\end{align*}
Here we used the fact that $u = k - v$. By substituting $i$ with $v - 1 - i$ we get
\begin{align*}
    \sum_{i = v-t}^{v-1}\binom{k-t-1}{i}(-1)^{v-1-i}\binom{t-1}{v-1-i} = \frac{(k-t-1)!}{(u-1)!(v-1)!}\sum_{i = 0}^{t-1}(-1)^i\binom{t-1}{i}\prod_{l = 1}^i\left(v - l\right)\prod_{l = 1}^{t-1-i}\left(u - l\right).
\end{align*}

We need to estimate the sum $\sum_{i = 0}^{t-1}(-1)^i\binom{t-1}{i}\prod_{l = 1}^i\left(v - l\right)\prod_{l = 1}^{t-1-i}\left(u - l\right)$. Since $u = \frac{1+\delta}{2}k$, $v = \frac{1-\delta}{2}k$ and $t$ is a constant, this sum is a polynomial in $k$ with degree $t-1$. The degree $t-1$ term is
\begin{align*}
    \sum_{i = 0}^{t-1}(-1)^i\binom{t-1}{i}v^iu^{t-1-i} = \left(u - v\right)^{t-1} = \delta^{t-1}k^{t-1}.
\end{align*}
Note that degree $t-2$ term is formed by taking one $l$ in one of the factors in $\prod_{l = 1}^i\left(v - l\right)\prod_{l = 1}^{t-1-i}\left(u - l\right)$ and taking $u$ or $v$ in the remaining factors, so it is
\begin{align*}
    &\sum_{i = 0}^{t-1}(-1)^i\binom{t-1}{i}\left(v^iu^{t-i-2}\sum_{l=1}^{t-1-i}(-l) + v^{i-1}u^{t-1-i}\sum_{l=1}^i(-l)\right) \\
    =\,\, & \sum_{i = 0}^{t-1}(-1)^i\binom{t-1}{i}\left(-v^iu^{t-i-2}\cdot\frac{(t-1-i)(t-i)}{2} - v^{i-1}u^{t-1-i}\cdot\frac{i(i+1)}{2}\right) \\
    =\,\, & \left((t-1)(u-v)^{t-2} - \frac{1}{2}(t-1)(t-2)v(u-v)^{t-3}\right) - \left((t-1)(u-v)^{t-2} + \frac{1}{2}(t-1)(t-2)u(u-v)^{t-3}\right) \\
    =\,\, & -\frac{1}{2}(t-1)(t-2)(u+v)(u-v)^{t-3} \\
    =\,\, & -\frac{1}{2}(t-1)(t-2)\delta^{t-3}k^{t-2}.
\end{align*}
Hence, the lemma follows.
\end{proof}

\section{Converting Non-integer Coefficient to Integer Coefficient}
In this section, we show that the assumption that the coefficient of the president is integer is not a serious restriction by proving the following theorem.
\begin{lemma}
Let $\delta$ be a constant with $0 < \delta < 1$. There exists a function $\delta' = \delta'(k)$ such that $\delta'\cdot k \in \mathbb{N}$ and
\[
\sign\left(\delta kx_1 + \sum_{i = 2}^k x_i\right) = \sign\left(\delta' kx_1 + \sum_{i = 2}^k x_i\right)
\]
for every $k \in \mathbb{N}$ and $x \in \{-1, 1\}^k$.
\end{lemma}
\begin{proof}
If $\delta k$ is already an integer, then let $\delta'(k) = \delta$. Otherwise, there exists $m \in \mathbb{N}$ such that $\delta k \in (m, m + 1)$. We will choose $\delta'(k) = m/k$ or $(m+1)/k$ depending on the following: if $k - 1$ is even, we choose $\delta'(k) \cdot k$ to be odd, and if $k - 1$ is odd, we choose $\delta'(k) \cdot k$ to be even. Note that the range of $\delta' kx_1 + \sum_{i = 2}^k x_i$ is 
\[
R = \{\delta' k + k - 1, \delta' k + k - 3, \ldots, \delta' k - k + 1, -\delta' k + k - 1, \delta' k + k - 3, \ldots, \delta' k - k + 1\}.
\]
By our choice of $\delta'$, every $t \in R$ is odd and therefore $|t| \geq 1$. For every $x \in \{-1, 1\}^k$, we have
\[
\left|\left(\delta kx_1 + \sum_{i = 2}^k x_i\right) - \left(\delta' kx_1 + \sum_{i = 2}^k x_i\right)\right| = \left| \delta k x_1 - \delta'k x_1\right | =  \left| \delta k - \delta'k\right | < 1,
\]
which implies that 
\[
\sign\left(\delta kx_1 + \sum_{i = 2}^k x_i\right) = \sign\left(\delta' kx_1 + \sum_{i = 2}^k x_i\right).
\]
\end{proof}
What we essentially did here is that we rounded $\delta k$ to either below or above. Note that this gives $|\delta'(k) - \delta| \leq 1/k$, which means when $k$ is sufficiently large, $\delta'$ and $\delta$ will be very close to each other. 
\end{document}